\documentclass[onecolumn]{IEEEtran}
\usepackage[T1]{fontenc} \usepackage[latin9]{inputenc}  \usepackage{array} \usepackage{float} \usepackage{multirow} 
 \usepackage{amsthm} \usepackage{amsmath}
 \usepackage{amssymb} \usepackage[normalem]{ulem} \usepackage{ulem}\usepackage[hyphens]{url}\usepackage{units}\usepackage{cite}\usepackage{algorithmic}\usepackage{mathrsfs}

\makeatletter

\providecommand{\tabularnewline}{\\} \floatstyle{ruled} \newfloat{algorithm}{tbp}{loa} \providecommand{\algorithmname}{Algorithm} \floatname{algorithm}{\protect\algorithmname}

\theoremstyle{plain} 
\newtheorem{thm}{\protect\theoremname} \theoremstyle{definition} 
\newtheorem{defn}[thm]{\protect\definitionname} \theoremstyle{plain} 
\newtheorem{lem}[thm]{\protect\lemmaname} \theoremstyle{plain} 
\newtheorem{prop}[thm]{\protect\propositionname} \theoremstyle{definition} 
\newtheorem{example}[thm]{\protect\examplename} \theoremstyle{remark} 
\newtheorem{rem}[thm]{\protect\remarkname}

\usepackage{graphicx}
\usepackage{subfig}

\makeatother

\usepackage[english]{babel} \providecommand{\definitionname}{Definition} \providecommand{\examplename}{Example} \providecommand{\lemmaname}{Lemma} \providecommand{\propositionname}{Proposition} \providecommand{\remarkname}{Remark} \providecommand{\theoremname}{Theorem}

\begin{document} 
\global\long\def\dist{\mathsf{d}} 
\global\long\def\kdist{K} 
\global\long\def\Kdist{\mathsf{d}_{K}} 
\global\long\def\Tdist{\mathsf{d}_{T}} 
\global\long\def\define{:\,=} 
\global\long\def\dash{\mbox{--}}
\global\long\def\adjset{A} 
\global\long\def\lineset{L}
\global\long\def\wtfn{\varphi} 
\global\long\def\wtf#1{\varphi_{#1}}
\newcommand{\nwt}[2]{\wtfn_{(#1\,#2)}}
\newcommand{\wt}{{\rm wt}}
\renewcommand{\S}{\mathbb S}
\newcommand{\Sn}{\mathbb S_n}
\hyphenation{ken-dall}
\newcommand{\ignore}[1]{{}}

\title{A Novel Distance-Based Approach to\\ Constrained Rank Aggregation}

\ignore{\author{
\alignauthor Farzad Farnoud (Hassanzadeh)\\
       \affaddr{University of Illinois at Urbana-Champaign}\\
       \email{hassanz1@illinois.edu}
\alignauthor Behrouz Touri\\
       \affaddr{University of Illinois at Urbana-Champaign}\\
       \email{touri1@illinois.edu}
\alignauthor Olgica Milenkovic
       \affaddr{University of Illinois at Urbana-Champaign}\\
       \email{milenkov@illinois.edu}
}
}

\author{Farzad~Farnoud~(Hassanzadeh), 
Olgica~Milenkovic, 
and~Behrouz~Touri \\
Department of Electrical and Computer Engineering, University of Illinois, Urbana-Champaign\\
E-mail: $\{{\text{hassanz1,milenkov,touri1}\}}$@illinois.edu
\thanks{This work was supported by the NSF grants CCF 0821910, CCF 0809895, CCF 0939370 and the AFRLDL-EBS AFOSR Complex Networks grant.  Part of the results were presented at SPCOM 2012, Bangalore, India and at ITA 2012, San Diego, CA.}}

\maketitle 
\begin{abstract} 
We consider a classical problem in choice theory -- vote aggregation -- using novel distance measures between permutations that arise in several practical applications.  The distance measures are derived through an axiomatic approach, taking into account various issues arising in voting with side constraints. The side constraints of interest include non-uniform relevance of the top and the bottom of rankings (or equivalently, eliminating negative outliers in votes) and similarities between candidates (or equivalently, introducing diversity in the voting process). The proposed distance functions may be seen as \emph{weighted versions} of the Kendall $\tau$ distance and \emph{weighted versions} of the Cayley distance. In addition to proposing the distance measures and providing the theoretical underpinnings for their applications,  we also consider algorithmic aspects associated with distance-based aggregation processes. We focus on two methods. One method is based on approximating weighted distance measures by a generalized version of Spearman's footrule distance, and it has provable constant approximation guarantees. The second class of algorithms is based on a non-uniform Markov chain method inspired by PageRank, for which currently only heuristic guarantees are known. We illustrate the performance of the proposed algorithms for a number of distance measures for which the optimal solution may be easily computed. 
\end{abstract}

\section{Introduction}

Rank aggregation, sometimes referred to as \emph{ordinal data fusion}, is a classical problem frequently encountered in the social sciences, web search and Internet service studies, expert opinion analysis, and economics~\cite{kemeney1959mathematics,cook1985ordinal,dwork2001rank-web,sculley2007rank,schalekamp2009rank,kumar2010gdr}. Rank aggregation plays a special role in information retrieval based on different 
search models, in cases when users initiate several queries for the information of interest to them,
or in situations when one has to combine various sources of evidence or use different document surrogates~\cite{infretrieval}.

The problem can be succinctly described as follows: a set of ``voters'' or ``experts'' is presented with a set of candidates (objects, individuals, movies, etc.). Each voter's task is to produce a \emph{ranking}, that is, an arrangement of the candidates in which the candidates are \emph{ranked} from the most preferred to the least preferred. The voters' rankings are then passed to an aggregator. The aggregator outputs a single ranking, 
termed the aggregate ranking, to be used as a representative of all voters.

Rank aggregation for votes including two candidates reduces to a simple majority count. The situation becomes significantly more complex when three or more candidates are considered. Two of the most obvious extensions of vote aggregation for two candidates to the case of more than two candidates are the majority rule and the Condorcet method (pairwise majority count). In the first case, one reduces the problem to counting how many times each candidate ended up at the top of the list. This candidate is declared the winner, and removed from all rankings. The same process is then performed to identify the second, etc., candidate in the list. In the second case, one aims at identifying the majority winner of pairwise competitions. Unfortunately, both methods are plagued by a number of issues that have cast doubt on the plausibility of fair vote aggregation. Examples include the famous Condorcet paradox~\cite{condorcet}, where pairwise comparisons may lead to intransitive results (i.e., for example, $a$ may be preferred to $b$, $b$ to $c$, and $c$ to $a$). 

To mitigate such problems, two other important categories of rank aggregation methods were studied in the past. These include \emph{score-based} methods and \emph{distance-based} methods. In score-based methods, the first variant of which was proposed by Borda~\cite{borda1784}, each candidate is assigned a score based on its position in each of the votes (rankings). The candidates are then ranked based on their total score. One argument in support of using Borda's count method is that it ranks highly those candidates supported at least to a certain extent by almost all voters, rather than candidates who are ranked highly only by the simple majority of voters.
In distance-based methods \cite{kemeney1959mathematics}, the aggregate is the deemed to be the ranking ``closest'' to the set of votes, or at the smallest cumulative distance from the votes, where closeness of two rankings is measured via some adequately chosen distance function. This approach can be thought of as finding the center of mass of the rankings, or the median -- centroid -- of the rankings, with the rankings representing point masses in a metric space. Well-known distance measures for rank aggregation include the Kendall $\tau$, the Cayley distance, and Spearman's Footrule \cite{diaconis1988group}. 

Clearly, the most important aspect of distance-based rank aggregation is to choose an appropriate distance function. One may argue that almost all problems arising in connection with the majority method or score based approaches directly translate into problems concerning the chosen distance measures. To address this issue, Kemeny \cite{kemeney1959mathematics, kemenySnell1962mathematical} presented a set of intuitively 
justifiable axioms that a distance measure must satisfy to be deemed suitable for aggregation purposes, and showed that only one distance measure satisfies the axioms -- namely, the Kendall $\tau$ distance. The Kendall $\tau$ distance between two rankings is the \emph{smallest} number of swaps of adjacent elements that transforms one ranking into the other. For example, the Kendall $\tau$ distance between the rankings $(1,3,4,2)$ and $(1,2,3,4)$ is two; we may first swap 2 and 4 and then 2 and 3 to transform $(1,3,4,2)$ to $(1,2,3,4)$. Besides its use in social choice and computer science theory, the Kendall $\tau$ distance has also received significant attention in the coding theory literature, due to its applications in modulation coding for flash memories~\cite{kendall1970rankcorrelation,bruck2009rank-modulation,6034261}.

Unfortunately, the Kendall $\tau$ is not a suitable distance measure for aggregation problems involving various electoral and Internet search engine constraints. Two such important constraints include differentiating the significance of the top versus the bottom of a ranking and differentiating candidates based on their ``similarity''.

In the first example, consider the following scenario. One may view the process of forming the aggregate ranking as one of ``tweaking'' a starting ranking so as to make it as close as possible to \emph{all} given voters' rankings. The effect of changing the ordering of candidates at the top or at the bottom is in principle the same -- i.e., if switching the top two elements in the aggregate reduces the total distance from the votes by the same amount as switching the bottom two elements, then both options are equally valid to be used. But in many applications, changes at, or near to, the top of rankings should not affect the distance between rankings to  same extent as changes at, or near to, the bottom of rankings. In other words, one should penalize making changes at the top of the list more than making changes at the bottom of the list, given that low ranked items are usually not very relevant. So far, only a handful of results are known for rank aggregation distances that address the problem of positional relevance, i.e., the significance of the top versus the bottom of rankings. One approach was described in~\cite{kumar2010gdr}, where the proposed distances were based on \emph{heuristic arguments} only. These approaches do not have axiomatic underpinnings, and efficient aggregation algorithms to accompany them are not known.

In the second example, consider a voting process were candidates should be ranked both based on merit and on a diversity criteria -- for example, not having more than two of the top ten candidates working in information theory. 
One may argue that in this case, using the Kendall $\tau$ distance for aggregation and reshuffling some candidates in order to satisfy the constraints, suffices to solve the problem. For example, one may move all except the two highest-ranked information theorists below position ten and leave the ranking unchanged otherwise. It is clear that this procedure may not be viewed as fair, since ranks of \emph{all} candidates were affected by the rankings of information theorists in the first place. Alternatively, one may reduce the search space only to rankings that satisfy the constraints, but this approach is computationally highly challenging.

Henceforth, we focus our attention on distance based aggregation methods catering to constraints of the form described above. The goal of our work is to provide an axiomatic underpinning for novel distance measures between rankings that take into account predetermined top-bottom and similarity/diversity constraints. In addition to their applications in computer science and social choice theory, 
these distance measures may be used in a variety of applications, ranging from bioinformatics to network analysis~\cite{farnoud2012sorting}. 

\subsection*{Motivation -- Top vs. Bottom}

Consider the ranking $\pi$ of the ``World's 10 best cities to live in'', according to a report composed by the Economist Intelligence Unit~\cite{eiuliveability}:
\begin{align*}
\pi = (&\text{Melbourne, Vienna, Vancouver, Toronto, Calgary,} \\
&\text{Adelaide, Sydney, Helsinki, Perth, Auckland})
\end{align*}

Now consider two other rankings that both differ from $\pi$ by one swap of adjacent entries:
\begin{align*} 
\pi' = (&\text{Melbourne, Vienna, Vancouver, Calgary, Toronto,}\\
&\text{Adelaide, Sydney, Helsinki, Perth, Auckland}),\\
\pi'' = (&\text{Vienna, Melbourne, Vancouver, Toronto, Calgary,}\\ 
&\text{Adelaide, Sydney, Helsinki, Perth, Auckland}).
\end{align*}
The astute reader probably immediately noticed that the top candidate was changed in $\pi''$, but otherwise took some time to realize where the adjacent swap appeared in $\pi'$. This is a consequence of the well-known fact that humans pay more attention to the top of the list rather than any other location in the ranking, and hence notice changes in higher positions easier\footnote{Note that one may argue that people are equally drawn to explore the highest and lowest ranked items in a list. For example, if about a hundred cities were ranked, it would be reasonable to assume that readers would be more interested in knowing the best and worst ten cities, rather than the cities occupying positions 41 to 60. These positional differences may also be addressed within the framework proposed in the paper.}. Note that the Kendall $\tau$ distance between $\pi$ and $\pi'$ and between $\pi$ and $\pi''$ is one, but it would appear reasonable to assume that the distance between $\pi$ and $\pi''$ be larger than that between $\pi$ and $\pi'$, as the corresponding swap occurred in a \emph{more significant} (higher ranked) position in the list.

\begin{figure}
\begin{center}
\includegraphics[width = 3.45in]{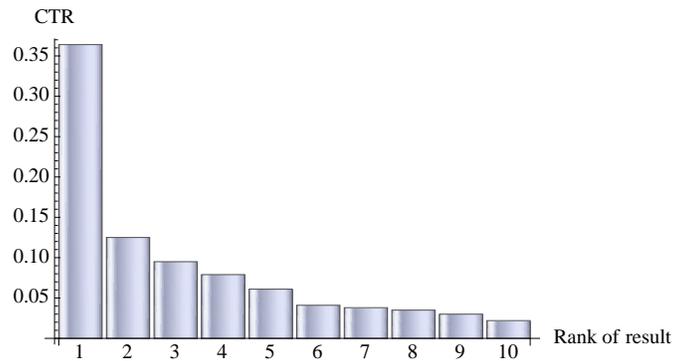}
\caption{Clickthrough rates (CTRs) of webpages appearing on the first page of Google search.}
\label{fig:serp}
\end{center}
\end{figure}

The second example corresponds to the well-studied notion of \emph{Clickthrough rates} (CTRs) of webpages in search engine results pages (SERPs). The CTR is used to assess the popularity of a webpage or the success rate of an online ad. It may be roughly defined as the number of times a link is clicked on divided by the total number of times that it appeared. A recent study by Optify Inc.~\cite{optifyinc} showed that the difference between the average CTR of the first (highest-ranked) result and the average CTR of the second (runner-up) result is very large, and much larger than the corresponding difference between the average CTRs of the lower ranked items (See Figure \ref{fig:serp}). Hence, in terms of directing search engine traffic, swapping higher-ranked adjacent pairs of search results has a larger effect on the performance of Internet services than swapping lower-ranked search results.

The aforementioned findings should be  considered when forming an aggregate ranking of webpages. For example, in studies of CTRs, one is often faced with questions regarding traffic flow from search engines to webpages. One may think of a set of keywords, each producing a different ranking of possible webpages, with the aggregate representing the median ranking based on different sets of keywords. Based on Figure \ref{fig:serp}, if
the ranking of a webpage is in the bottom half, its exact position is not as relevant as when it is ranked in the top half. Furthermore, a webpage appearing roughly half of the time at the top and roughly half of the time at the
bottom will generate more incoming traffic than a webpage with persistent average ranking.

Throughout the paper, we refer to the above-mentioned problem as the ``top-vs-bottom'' problem. Besides the importance in emphasizing the relevance of the top of the list, distance measures that penalize perturbations at the top of the list more than perturbations at the bottom of the list have another important application in practice -- to eliminate negative outliers. As will be shown in subsequent sections, top-vs-bottom distance measures allow candidates to be highly ranked in the aggregate even though they have a certain (small) number of highly negative ratings. The policy of eliminating outliers before rating items or individuals is a well-known one, but has not been considered in the social choice literature in the context of distance-based rank aggregation.

\subsection*{Motivation -- Similarity of Candidates}

In many vote aggregation problems, the identity of the candidates may not be known. On the other hand, many other applications require that the identity of the candidates be revealed. In this case, candidates are frequently partitioned in terms of some similarity criteria -- for example, area of expertise, gender, working hour schedule etc. Hence, pairs of candidates may have different degrees of similarity and swapping candidates that are similar should be penalized less than swapping candidates that are not similar according to the given ranking criteria. For example, in a faculty search ranking one may want to have at least one but not more than two physicists ranked among the top 10 candidates, or at least two women among the top 5 candidates.

Pertaining to the Economist Intelligence Unit ranking, one may also consider the identity of the elements that are swapped, and not only their position. In this case, it may be observed that the swap in $\pi''$ involves cities on two different continents, which may shift the general opinion about the cities' countries of origin. On the other hand, the two cities swapped in $\pi'$ are both in Canada, so that the swap is not likely to change the perception of quality of living in that country. This points to the need for distance measures that take into account similarities and dissimilarities among candidates. 

Distance measures capable of integrating these criteria directly into the aggregation process are not know in the literature. A class of distance measures introduced by the authors in~\cite{farnoud2012sorting}, termed weighted transposition distance, is suitable for this task as it can take into account similarities and dissimilarities of candidates. The weighted transposition distance can be viewed as a generalization of both the previously described Kendall $\tau$ and the so called \emph{Cayley distance} between permutations. The Cayley distance between two rankings is the \emph{smallest} number of (not necessarily adjacent) swaps required to transform one ranking into the other. For example, the Cayley distance between the permutations $(1,2,3,4)$ and $(1,4,3,2)$ is one, since the former can be transformed into the latter by swapping the elements $2$ and $4$. Note that while the Cayley distance allows for arbitrary swaps, the Kendall $\tau$ distance allows for swaps of adjacent elements only. 
It is straightforward to see that every statements made about swapping elements $i$ and $j$ may be converted to statements made about swapping elements at positions $i$ and $j$ by using the inverse of the ranking/permutation. Similarity of items may be captured by assigning costs or weights to swaps, and choosing the transposition weights so that swapping dissimilar items induces a higher weight/distance compared to swapping 
similar items. This approach is the topic of the next two sections. 

To address the top-vs-bottom and similarity issues, we axiomatically describe a class of distance functions by assigning different weights to different adjacent and non-adjacent swaps, termed the weighted Kendall and weighted transposition distance measures, respectively. Furthermore, we show that the proposed distance functions can be computed in polynomial time in some special cases and provide a polynomial-time 2-approximation algorithm for the general case. The results we present also pertain to algorithmic aspects of rank aggregation~\cite{kemeney1959mathematics,cook1985ordinal,sculley2007rank,schalekamp2009rank}. 

In this setting, we describe the performance of an algorithm for rank aggregation based on a generalization of Spearman's footrule distance and solving a minimum weight matching problem (this algorithm is inspired by a procedure described in~\cite{dwork2001rank-web}) and a combination of the matching algorithm with local descent methods. Furthermore, we describe an algorithm reminiscent of PageRank~\cite{dwork2001rank-web}, where the ``hyperlink probabilities'' are chosen according to swapping likelihoods (weights).

The remainder of the paper is organized as follows. An overview of relevant concepts, definitions, and terminology is presented in Section \ref{sec:preliminaries}. Weighted Kendall distance measures, as well their axiomatic definitions, are presented in Section \ref{sec:wkd}. Section \ref{sec:weighted-cayley} is devoted to the weighted transposition distance and its computational aspects. Novel rank aggregation algorithms for the weighted Kendall and weighted transposition distances are presented in Section \ref{sec:alg}.

\section{Preliminaries} \label{sec:preliminaries}

Formally, a ranking is a list of candidates arranged in order of preference, with the first candidate being the most preferred and the last candidate being the least preferred one. 

Consider the set of all possible rankings of a set of $n$ candidates. Via an arbitrary, but fixed, injective mapping from the set of candidates to $\{1,2,\cdots,n\}=[n]$, each ranking may be represented as a permutation. The mapping is often implicit and we usually equate rankings of $n$ candidates with permutations in $\S_n$, where $\S_n$ denotes the symmetric group of order $n$. This is equivalent to assuming that the set of candidates is the set $[n]$. For notational convenience, we use Greek lower-case letters for permutations, and explicitly write permutations as ordered sets $\sigma=(\sigma(1),\ldots,\sigma(n))$.

Let $e$ denote the identity permutation $(1,2,\cdots,n)$. For two permutations $\pi,\sigma\in \S_n$, the product $\mu=\pi\sigma$ is defined via the identity $\mu(i) = \pi(\sigma(i)), i=1,2,\cdots n$.

\begin{defn} A \emph{transposition} $\tau = (a\,b)$, for $a,b\in [n]$ and $a\neq b$, is a permutation that swaps $a$ and $b$ and keeps all other elements of $e$ fixed. That is,
\[\tau(i)=\begin{cases}
b, & \qquad i=a,\\
a, & \qquad i=b,\\
i, & \qquad\mbox{else}
\end{cases}\]
If $|a-b|=1$, the transposition is referred to as an \emph{adjacent transposition.} \end{defn} 
Note that for $\pi\in\mathbb{S}_{n}$, $\pi\left(a\,b\right)$ is obtained from $\pi$ by swapping elements in positions $a$ and $b$, and $\left(a\,b\right)\pi$ is obtained by swapping $a$ and $b$ in $\pi$. For example, $(3,1,4,2)(2\,3) = (3,4,1,2)$ and $(2\,3)(3,1,4,2) = (2,1,4,3)$.

For our future analysis, we define the set 
\begin{align*} 
\adjset(\pi,\sigma)=&\bigl\{\tau=\left(\tau_{1},\cdots,\tau_{|\tau|}\right): \\
&\sigma=\pi\tau_{1}\cdots\tau_{|\tau|},\tau_{i}=\left(a_{i}\,a_{i}+1\right),i\in[|\tau|]\bigr\} 
\end{align*} 
i.e., the set of all ordered sequences of adjacent transpositions that \emph{transform} $\pi$ into $\sigma.$ The fact that $A(\pi,\sigma)$ is non-empty, for any $\pi,\sigma\in \S_n$, is obvious. Using $\adjset(\pi,\sigma)$, the Kendall $\tau$ distance between two permutations $\pi$ and $\sigma$, denoted by $\kdist(\pi,\sigma)$, may be written as 
\[\kdist(\pi,\sigma) = \min_{\tau \in A(\pi,\sigma)} |\tau|.\]

For a ranking $\pi\in\mathbb{S}_{n}$ and $a,b\in[n]$, $\pi$ is said to \emph{rank} $a$ \emph{before} $b$ or \emph{higher than} $b$ if $\pi^{-1}(a)<\pi^{-1}(b)$. We denote this relationship as $a<_{\pi}b$. Two rankings $\pi$ and $\sigma$ \emph{agree} on the relative order of a pair $\{a,b\}$ of elements if both rank $a$ before $b$ or both rank $b$ before $a$. Furthermore, the two rankings $\pi$ and $\sigma$ \emph{disagree} on the relative order of a pair $\{a,b\}$ if one ranks $a$ before $b$ and the other ranks $b$ before $a$. For example, consider $\pi=(1,2,3,4)$ and $\sigma=(4,2,1,3)$. We have that $4<_{\sigma}1$ and that $\pi$ and $\sigma$ agree on $\{2,3\}$ but disagree on $\{1,2\}$. 

Given a distance function $\dist$ over the permutations in $\mathbb{S}_{n}$ and a set $\Sigma=\{\sigma_1,\cdots,\sigma_m\}$ of $m$ votes (rankings), the distance-based aggregation problem can be stated as follows: find the ranking $\pi^{*}$ that minimizes the cumulative distance from $\Sigma$, i.e., \begin{equation} \pi^{*}=\arg\min_{\pi\in\mathbb{S}_{n}}\sum_{i=1}^{m}\dist(\pi,\sigma_{i}).\label{eqn:rank-agg} \end{equation}
In words, the goal is to find a ranking $\pi$ that represents the median of the set of permutations $\Sigma$. The choice of the distance function $\dist$ is an important aspect of distance-based rank aggregation and the focus of the paper.

In \cite{kemeney1959mathematics}, Kemeny presented a set of axioms that a distance function for rank aggregation should satisfy and proved that the only distance that satisfies the axioms is the Kendall $\tau$. A critical concept in Kemeny's axioms is the idea of ``betweenness,'' defined below.

\begin{defn} \label{def:between} A ranking $\omega$ is said to be \emph{between} two rankings $\pi$ and $\sigma$, denoted by $\pi\dash\omega\dash\sigma$, if for each pair of elements $\{a,b\}$, $\omega$ either agrees with $\pi$ or $\sigma$ or both. The rankings $\pi_{1},\pi_2,\cdots,\pi_{s}$ are said to be \emph{on a line}, denoted by $\pi_{1}\dash\pi_{2}\dash\cdots\dash\pi_{s}$, if for every $i,j,$ and $k$ for which $1\le i<j<k\le s$, we have $\pi_{i}\dash\pi_{j}\dash\pi_{k}$. 
\end{defn} 

In Kemeny's work, rankings are allowed to have ties. The basis of our subsequent analysis is the same set of axioms, listed below. 
However, our focus is on ranking without ties, in other words, permutations.

\vspace{1mm}
\textbf{Axioms I} 
\begin{enumerate} 
\item $\dist$ is a metric. 
\item $\dist$ is left-invariant. 
\item For any $\pi,\sigma,$ and $\omega$, $\dist(\pi,\sigma)=\dist(\pi,\omega)+\dist(\omega,\sigma)$ if and only if $\omega$ is between $\pi$ and $\sigma$. 
\item The smallest positive distance is one. 
\end{enumerate}

Axiom 2 states that relabeling of objects should not change the distance between permutations. In other words, $\dist(\sigma\pi,\sigma\omega)=\dist(\pi,\omega)$, for any $\pi,\sigma,\omega\in\mathbb{S}_{n}$. Axiom 3 may be viewed through a geometric lens: the triangle inequality has to be satisfied with equality for all points that lie on a line between $\pi$ and $\sigma$. Axiom 4 is only used for normalization purposes.

Kemeny's original exposition included a fifth axiom which we state for completeness: If two rankings $\pi$ and $\sigma$ agree except for a segment of $k$ elements, the position of the segment does not affect the distance between the rankings. Here, a segment represents a set of objects that are ranked consecutively -- i.e., a substring of the permutation. As an example, this axiom implies that 
\begin{align*}\dist((1,2,3,\underbrace{4,5,6}),&(1,2,3,\underbrace{6,5,4}))
=\\
\dist(&(1,\underbrace{4,5,6},2,3),(1,\underbrace{6,5,4},2,3)) \end{align*} 
where the segment is underscored by braces. This axiom clearly enforces a property that \emph{is not desirable} for metrics designed to address the top-vs-bottom issue: changing the position of the segment in two permutations does not alter their mutual distance. One may hence believe that removing this axiom (as was done in Axioms I) will lead to distance measures capable of handling the top-vs-bottom problem.
But as we show below, for rankings without ties, omitting this axiom does not change the outcome of Kemeny's analysis. In other words, the axiom is redundant. This is a rather surprising fact, and we conjecture that the
same is true of rankings with ties. 

In the remainder of this section, we demonstrate the redundancy of Kemeny's fifth axiom and use our novel proof method to identify how to change the axioms in Axioms I in order to arrive at 
distance measures that cater to the need of top-vs-bottom and similarity problems. For reasons that will become clear in the next section, we refer to distance measures resulting from such axioms as \emph{weighted distances}.

The main result of this section is Theorem \ref{thm:unique-dist}, stating that the unique distance satisfying Axioms I is the Kendall $\tau$ distance. The theorem is proved with the help of Lemmas \ref{lem:on-the-line}, \ref{lem:two-ways}, \ref{lem:gamma-1}, \ref{lem:inversions} and \ref{lem:kendal-inversions}. 

\begin{lem} \label{lem:on-the-line}
For any distance measure $\dist$ that satisfies Axioms I, and for any sequence of permutations $\pi_{1},\pi_2,\cdots,\pi_{s}$ such that $\pi_{1}\dash\pi_{2}\dash\cdots\dash\pi_{s}$, one has \[ \dist(\pi_{1},\pi_{s})=\sum_{k=1}^{s-1}\dist(\pi_{k},\pi_{k+1}). \] 
\end{lem} 

\begin{proof} 
The lemma follows from Axiom I.3 by induction.\end{proof} 

\begin{lem} \label{lem:two-ways}
For any $\dist$ that satisfies Axioms I and for $i\in[n-1]$, we have 
\[ \dist\left(\left(i\,i+1\right),e\right)=\dist\left((1\,2),e\right). \] 
\end{lem} 

\begin{proof} We first show that $\dist\left(\left(2\,3\right),e\right)=\dist\left(\left(1\,2\right),e\right)$. Repeating the same argument used for proving this special case gives $\dist\left(\left(i\,i+1\right),e\right)=\dist\left(\left(i-1\,i\right),e\right)=\cdots=\dist\left(\left(1\,2\right),e\right)$.

To show that $\dist\left(\left(2\,3\right),e\right)=\dist\left(\left(1\,2\right),e\right)$, we evaluate $\dist(\pi,e)$ in two ways, where we choose $\pi=(3,2,1,4,5,\cdots, n).$

On the one hand, note that $\pi\dash\omega\dash\eta\dash e$, where $\omega=\pi(1\,2)=(2,3,1,4,5,\cdots, n)$ and $\eta=\omega(2\,3)=(2,1,3,4,5,\cdots, n)$. As a result, \begin{align} \dist(\pi,e) & =\dist(\pi,\omega)+\dist(\omega,\eta)+\dist(\eta,e)\nonumber \\ & =\dist(\omega^{-1}\pi,e)+\dist(\eta^{-1}\omega,e)+\dist(\eta,e)\nonumber \\ & =\dist((1\,2),e)+\dist((2\,3),e)+\dist((1\,2),e)\label{eq:one-hand} \end{align} where the first equality follows from Lemma \ref{lem:on-the-line}, while the second is a consequence of the left-invariance property of the distance measure.

On the other hand, note that $\pi\dash\alpha\dash\beta\dash e$, where $\alpha=\pi(2\,3)=(3,1,2,4,5,\cdots, n)$ and $\beta=\alpha(1\,2)=(1,3,2,4,5,\cdots, n)$. For this case, \begin{align} \dist(\pi,e) & =\dist(\pi,\alpha)+\dist(\alpha,\beta)+\dist(\beta,e)\nonumber \\ & =\dist(\alpha^{-1}\pi,e)+\dist(\beta^{-1}\alpha,e)+\dist(\beta,e)\nonumber \\ & =\dist((2\,3),e)+\dist((1\,2),e)+\dist((2\,3),e).\label{eq:other-hand} \end{align}

Equations (\ref{eq:one-hand}) and (\ref{eq:other-hand}) imply that $\dist\left(\left(2\,3\right),e\right)=\dist\left(\left(1\,2\right),e\right)$.
\end{proof} 

\begin{lem} \label{lem:gamma-1}For any $\dist$ that satisfies Axioms I, $\dist(\gamma,e)$ equals the minimum number of adjacent transpositions required to transform $\gamma$ into $e$.\end{lem} 

\begin{proof} Let 
\begin{align*} \lineset(\pi,\sigma)=\{\tau=\left(\tau_{1},\cdots,\tau_{|\tau|}\right)&\in\adjset(\pi,\sigma):\\ \pi\dash\pi\tau_{1}\dash\pi\tau_{1}\tau_{2}\dash\cdots\dash\sigma\} 
\end{align*} 
be the subset of $\adjset(\pi,\sigma)$ consisting of sequences of transpositions that transform $\pi$ into $\sigma$ by passing through a line. Let $s$ be the minimum number of adjacent transpositions that transform $\gamma$ into $e$. Furthermore, let $(\tau_{1},\tau_{2},\cdots,\tau_{s})\in A(\gamma,e)$ and define $\gamma_{i}=\gamma\tau_{1}\cdots\tau_{i},i=0,\cdots,s,$ with $\gamma_{0}=\gamma$ and $\gamma_{s}=e$.

First, we show $\gamma_{0}\dash\gamma_{1}\dash\cdots\dash\gamma_{s}$, that is, 
\begin{equation} 
(\tau_{1},\tau_{2},\cdots,\tau_{s})\in L(\gamma,e).\label{eq:linear-trans} 
\end{equation} 
Suppose this were not the case. Then, there exist $i<j<k$ such that $\gamma_{i},\gamma_{j},$ and $\gamma_{k}$ are not on a line, and thus, there exists a pair $\{r,s\}$ for which $\gamma_{j}$ disagrees with both $\gamma_{i}$ and $\gamma_{k}$. Hence, there exist two transpositions, $\tau_{i'}$ and $\tau_{j'}$, with $i<i'\le j$ and $j<j'\le k$ that swap $r$ and $s$. We can in this case remove $\tau_{i'}$ and $\tau_{j'}$ from $\left(\tau_{1},\cdots,\tau_{s}\right)$ to obtain $\left(\tau_{1},\cdots,\tau_{i'-1},\tau_{i'+1},\cdots,\tau_{j'-1},\tau_{j'+1},\tau_{s}\right)\in A(\gamma,e)$ with length $s-2$. This contradicts the optimality of the choice of $s$. Hence, $(\tau_{1},\tau_{2},\cdots,\tau_{s})\in L(\gamma,e)$. Then Lemma \ref{lem:on-the-line} implies that 
\begin{equation} \dist(\gamma,e)=\sum_{i=1}^{s}\dist(\tau_{i},e).\label{eq:A-1} \end{equation}

Lemma \ref{lem:two-ways} states that all adjacent transpositions have the same distance from the identity. Since transpositions $\tau_i, 1\le i\le s,$ in \eqref{eq:A-1} are adjacent transpositions, $\dist(\tau_i,e)=a$ for some $a>0$ and thus $\dist(\gamma,e) = sa$. 

In (\ref{eq:A-1}), the minimum positive distance is obtained when $s=1$. That is, the minimum positive distance from identity equals $a$ and is obtained when $\gamma$ is an adjacent transposition. Axiom I.4 states that the minimum positive distance is 1. By left-invariance, this axiom implies that the minimum positive distance of any permutation from the identity is 1. Hence, $a=1$ and for any $\gamma\in\Sn$, \[\dist(\gamma,e)=\sum_{i=1}^{s}\dist(\tau_{i},e)=sa=s. \] \end{proof} 

\begin{lem} \label{lem:inversions}For any $\dist$ that satisfies Axioms I, and for $\pi,\sigma \in \mathbb{S}_n,$ we have
\[ \dist(\pi,\sigma)=\min\left\{ s:(\tau_{1},\cdots,\tau_{s})\in A(\pi,\sigma)\right\} . \] 
\end{lem} 

\begin{proof} We have $(\tau_{1},\cdots,\tau_{s})\in A(\pi,\sigma)$ if and only if \[(\tau_{1},\cdots,\tau_{s})\in A(\sigma^{-1}\pi,e).\] 
Left-invariance of $\dist$ implies that $\dist(\pi,\sigma)=\dist(\sigma^{-1}\pi,e)$. Hence, \begin{align*} \dist(\pi,\sigma) & =\dist(\sigma^{-1}\pi,e)\\ & =\min\left\{ s:(\tau_{1},\cdots,\tau_{s})\in A(\sigma^{-1}\pi,e)\right\} \\ & =\min\left\{ s:(\tau_{1},\cdots,\tau_{s})\in A(\pi,\sigma)\right\} \end{align*} where the second equality follows from Lemma \ref{lem:gamma-1}.\end{proof} 

For $\pi,\sigma\in\mathbb{S}_{n}$, let \[ I\left(\pi,\sigma\right)=\left\{ \{i,j\}:i<_{\pi}j,j<_{\sigma}i\right\} \] be the set of pairs $\{i,j\}$ on which $\pi$ and $\sigma$ disagree. The number $|I(\pi,\sigma)|$ is usually referred to as the  number of \emph{inversions} between the two permutations. 

The following lemma show that the Kendall $\tau$ distance between a permutation $\pi$ and $e$ equals the number of inversions in $\pi$. The result of the lemma is known, but a sketch of a proof is provided  for completeness.

 \begin{lem} \label{lem:kendal-inversions}For $\pi,\sigma\in\mathbb{S}_{n}$, \[ \kdist(\pi,\sigma)=\left|I(\pi,\sigma)\right|. \] \end{lem} \begin{proof} Consider a sequence $\tau_{1},\cdots,\tau_{k}$ of adjacent transpositions that transforms $\pi$ into $\sigma$, i.e., $\sigma=\pi\tau_{1}\cdots\tau_{k}$, with $k=\kdist(\pi,\sigma)$. Let $\pi_{j}=\pi\tau_{1}\cdots\tau_{j}$. Each $\tau_{i}$ decreases the number of inversions by at most one. Hence, \begin{align*} \left|I(\pi_{j},\sigma)\right| & \ge\left|I(\pi_{j-1},\sigma)\right|-1 \end{align*} and thus \[ 0=\left|I(\pi_{k},\sigma)\right|\ge\left|I(\pi,\sigma)\right|-k. \] Since $k=\kdist(\pi,\sigma)$, we obtain \[ \kdist(\pi,\sigma)\ge\left|I(\pi,\sigma)\right|. \]

On the other hand, it is easy to see that one can find $\tau_{i},i\in[k]$ in such a way that each $\tau_i$ decreases the number of inversions by one. For example, Bubble Sort~\cite{cormen24introduction} is one such well-known algorithm for accomplishing this task. Hence, 
\[ \kdist(\pi,\sigma)=\left|I(\pi,\sigma)\right|. \]
\end{proof}

\begin{thm} \label{thm:unique-dist}The unique distance $\dist$ that satisfies Axioms I is \[ \kdist(\pi,\sigma)=\min\left\{ s:(\tau_{1},\cdots,\tau_{s})\in A(\pi,\sigma)\right\} . \] \end{thm} \begin{proof} We show below that $\kdist$ satisfies Axiom I.3, as proving that $\kdist$ satisfies the other axioms is straightforward. Uniqueness follows from Lemma \ref{lem:inversions}.

 To show that $\kdist$ satisfies Axiom I.3, we use Lemma \ref{lem:kendal-inversions} stating that \[ \kdist(\pi,\sigma)=\left|I(\pi,\sigma)\right|. \] Fix $\pi,\sigma\in\mathbb{S}_{n}$. For any $\omega\in\mathbb{S}_{n}$, it is clear that \begin{equation} I(\pi,\sigma)\subseteq I(\pi,\omega)\cup I(\omega,\sigma).\label{eq:in-union} \end{equation}

Suppose first that $\omega$ is not between $\pi$ and $\sigma$. Then there exists a pair $\left\{ a,b\right\} $ with $a<_{\pi}b$ and $a<_{\sigma}b$ but with $a>_{\omega}b$. Since $\{a,b\}\notin I(\pi,\sigma)$ but $\{a,b\}\in I(\pi,\omega)\cup I(\omega,\sigma)$, we find that 
\[ \left|I(\pi,\sigma)\right|<\left|I(\pi,\omega)\cup I(\omega,\sigma)\right| ,\] 
and thus \begin{align*} \kdist(\pi,\sigma) & =\left|I(\pi,\sigma)\right|\\ & <\left|I(\pi,\omega)\cup I(\omega,\sigma)\right|\\ & \le\left|I(\pi,\omega)\right|+\left|I(\omega,\sigma)\right|\\ & =\kdist(\pi,\omega)+\kdist(\omega,\sigma). \end{align*} 

Hence, if $\omega$ is not between $\pi$ and $\sigma$, then \[ \kdist(\pi,\sigma)\neq\kdist(\pi,\omega)+\kdist(\omega,\sigma). \] 

Next, suppose $\omega$ is between $\pi$ and $\sigma$. This immediately implies that $I(\pi,\omega)\subseteq I(\pi,\sigma)$ and $I(\omega,\sigma)\subseteq I(\pi,\sigma)$. These relations, along with (\ref{eq:in-union}) imply that \begin{equation} I(\pi,\omega)\cup I(\omega,\sigma)=I(\pi,\sigma).\label{eq:union-in} \end{equation}

We claim that $I(\pi,\omega)\cap I(\omega,\sigma)=\emptyset$. To see this, observe that if $\{a,b\}\in I(\pi,\omega)\cap I(\omega,\sigma)$, then the relative rankings of $a$ and $b$ are the same for $\pi$ and $\sigma$ and so, $\{a,b\}\notin I(\pi,\sigma)$. The last statement contradicts \eqref{eq:union-in} and thus \begin{equation} I(\pi,\omega)\cap I(\omega,\sigma)=\emptyset.\label{eq:intersection} \end{equation} 

From (\ref{eq:union-in}) and (\ref{eq:intersection}), we may write \begin{align*} \kdist(\pi,\sigma) & =\left|I(\pi,\sigma)\right|\\ & =\left|I(\pi,\omega)\cup I(\omega,\sigma)\right|\\ & =\left|I(\pi,\omega)\right|+\left|I(\omega,\sigma)\right|\\ & =\dist(\pi,\omega)+\dist(\omega,\sigma), \end{align*} and this completes the proof of the fact that $\kdist$ satisfies Axiom I.3. \end{proof} 


A distance $\dist$ over $\mathbb{S}_{n}$ is called a \emph{graphic distance} \cite{deza1998metrics} if there exists a graph $G$ with vertex set $\mathbb{S}_{n}$ such that 
for $\pi,\sigma\in\mathbb{S}_{n},$ $\dist\left(\pi,\sigma\right)$ is equal to the length of the shortest path between $\pi$ and $\sigma$ in $G$. Note that this definition implies that the edge set of $G$ is the set \[ \left\{ \left(\alpha,\beta\right):\alpha,\beta\in\mathbb{S}_{n},\dist\left(\alpha,\beta\right)=1\right\} . \]

The Kendall $\tau$ distance is a graphic distance. To see the validity of this claim, take the corresponding graph to have vertices indexed by permutations, with an edge between each pair of permutations that differ by only one adjacent transposition. 

In the next section, we introduce the weighted Kendall distance which may be viewed as the shortest path between permutations over a \emph{weighted graph,} and show how this distance arises from modifying Kemeny's axioms.

\begin{figure*} \label{fig:GraphicKendall}
\begin{center}
\subfloat[]{\includegraphics[width=2.5in]{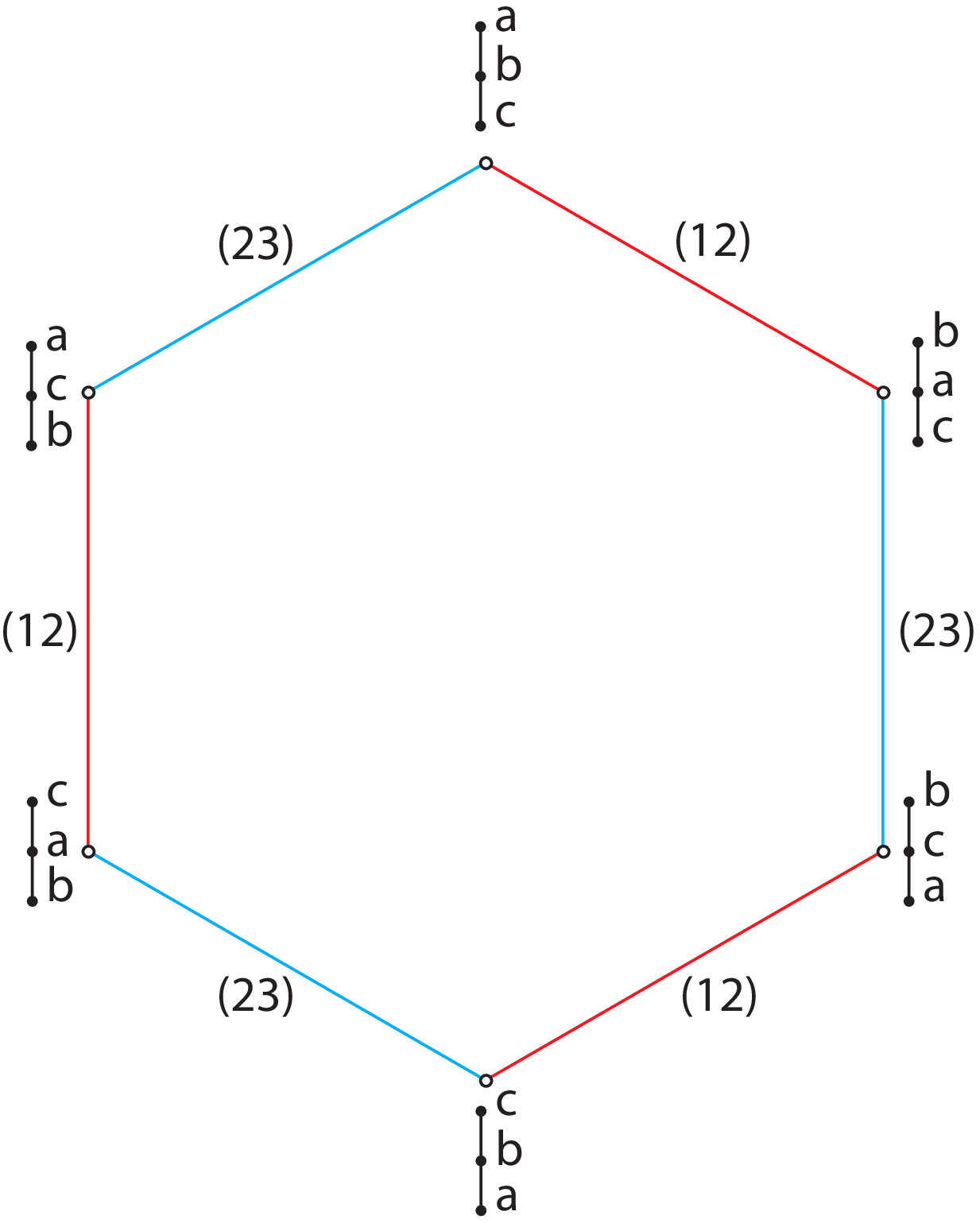}}\hspace{0.2in}
\subfloat[]{\includegraphics[width=2.5in]{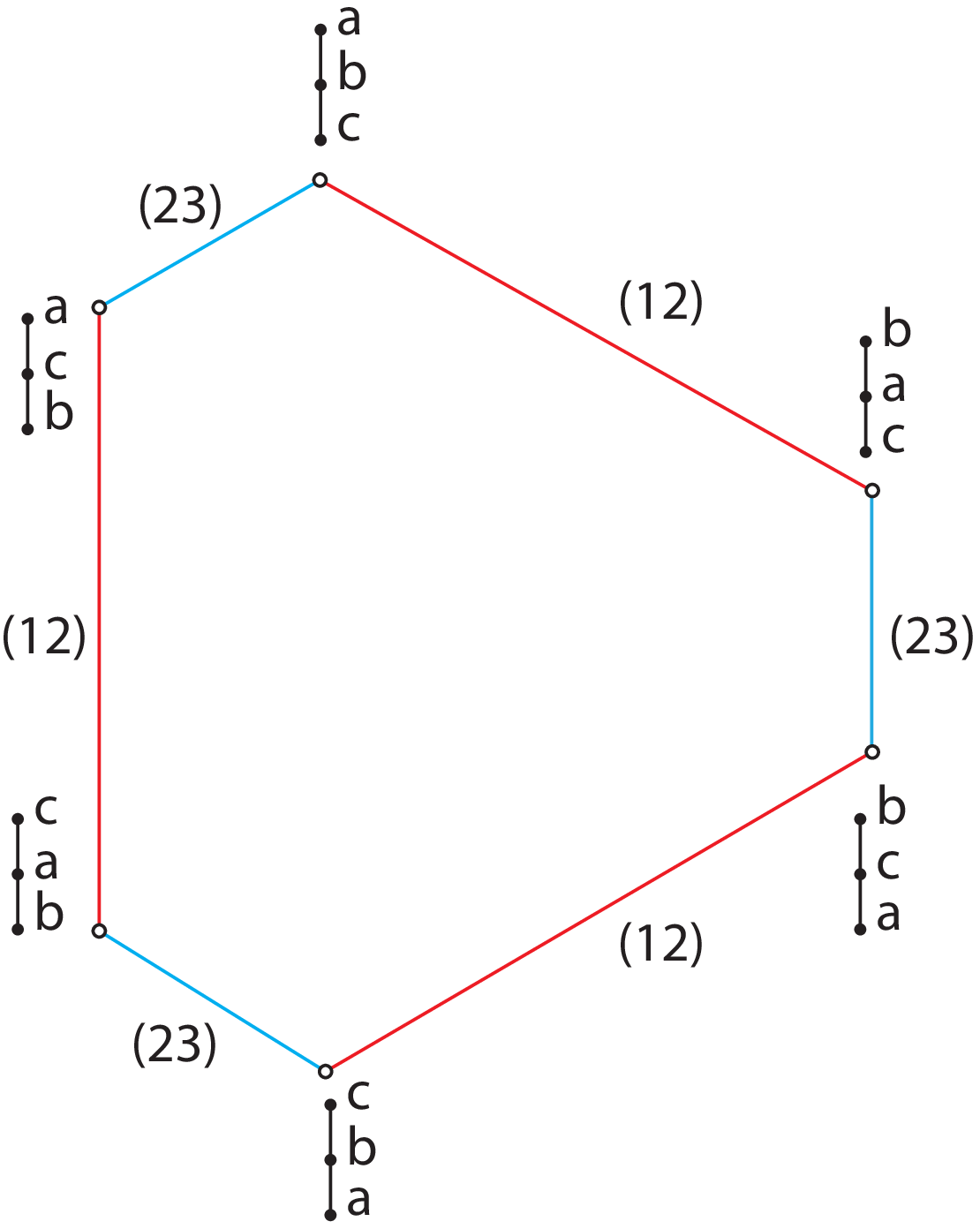}}
\caption{The graphs for the Kendall $\tau$ distance (a), and weighted Kendall distance (b)}
\end{center}
\end{figure*}

\section{The Weighted Kendall Distance} \label{sec:wkd}

The proof of the uniqueness of the Kendall $\tau$ distance under Axioms I reveals an important insight: the Kendall $\tau$ distance arises due to the fact that adjacent transpositions have uniform costs, which is a consequence of the betweenness property described in one of the axioms. If one had a ranking problem in which weights of transpositions either depended on the identity of the elements involved in the transposition or their positions, the uniformity assumption would have to be changed. As we show below, a way to achieve this goal is to redefine the axioms in terms of the betweenness property.

\vspace{1mm}
 \textbf{Axioms II} \begin{enumerate} 
\item $\dist$ is a pseudo-metric, i.e. a generalized metric in which two distinct points may be at distance zero.
\item $\dist$ is left-invariant. 
\item For any $\pi,\sigma$ disagreeing on more than one pair of elements, there exists \emph{some} $\omega$, distinct from $\pi$ and $\sigma$ and between them, such that $\dist(\pi,\sigma)=\dist(\pi,\omega)+\dist(\omega,\sigma)$.
\end{enumerate} %
Axiom II.1 allows for the option that some transpositions are not penalized or counted, due to the side constraints of the voting process. Intuitively, Axiom II.3 states that there exists at least one point on some line between $\pi$ and $\sigma$, for which the triangle inequality is an equality. In other words, there exists one ``shortest line'' between two permutations, and not all straight lines are required to be of the same length (see Figure 2 for an illustration). 

\begin{lem} \label{lem:WKD}For any distance $\dist$ that satisfies Axioms II, and for distinct $\pi$ and $\sigma$, we have \[ \dist(\pi,\sigma)=\min_{\left(\tau_{1},\cdots,\tau_{s}\right)\in\adjset(\pi,\sigma)}\sum_{i=1}^{s}\dist(\tau_{i},e). \] \end{lem} 

\begin{proof} The proof follows by induction on $\kdist(\pi,\sigma)$, the Kendall $\tau$ distance between $\pi$ and $\sigma$. 

First, suppose that $\kdist(\pi,\sigma)=1$, i.e., $\pi$ and $\sigma$ disagree on one pair of adjacent elements. Then, we have $\sigma=\pi(a\,a+1)$ for some $a\in[n-1]$. For each $\left(\tau_{1},\cdots,\tau_{s}\right)\in\adjset(\pi,\sigma)$, there exists an index $j$ such that $\tau_{j}=(a\,a+1)$ and thus \[ \sum_{i=1}^{s}\dist(\tau_{i},e)\ge\dist(\tau_{j},e)=\dist(\left(a\,a+1\right),e) \] implying \begin{equation} \min_{\left(\tau_{1},\cdots,\tau_{s}\right)\in A(\pi,\sigma)}\sum_{i=1}^{s}\dist(\tau_{i},e)\ge\dist(\left(a\,a+1\right),e).\label{eq:lq} \end{equation} On the other hand, since $\left((a\,a+1)\right)\in\adjset(\pi,\sigma)$, \begin{equation} \min_{\left(\tau_{1},\cdots,\tau_{s}\right)\in A(\pi,\sigma)}\sum_{i=1}^{s}\dist(\tau_{i},e)\le\dist(\left(a\,a+1\right),e).\label{eq:hq} \end{equation} From (\ref{eq:lq}) and (\ref{eq:hq}), \[ \dist(\pi,\sigma)=\min_{\left(\tau_{1},\cdots,\tau_{s}\right)\in A(\pi,\sigma)}\sum_{i=1}^{s}\dist(\tau_{i},e)=\dist(\left(a\,a+1\right),e)\] where the last equality follows from the left-invariance of $\dist$.

Next, suppose that $\kdist(\pi,\sigma)>1$, i.e., $\pi$ and $\sigma$ disagree on more than one pair of adjacent elements, and that for all $\mu,\eta\in\S_n$ with $\kdist(\mu,\eta)<\kdist(\pi,\sigma)$, the lemma holds. Then, there exists $\omega$, distinct from $\pi$ and $\sigma$ and between them, such that
\begin{equation*}
\begin{split}
\dist(\pi,\sigma) &= \dist(\pi,\omega)+\dist(\omega,\sigma),\\
\kdist(\pi,\omega) &< \kdist(\pi,\sigma),\\
\kdist(\omega,\sigma)&<\kdist(\pi,\sigma).
\end{split}
\end{equation*}

By the induction hypothesis, there exist $\left(\nu_{1},\cdots,\nu_{k}\right)\in A(\pi,\omega)$ and $\left(\nu_{k+1},\cdots,\nu_{s}\right)\in A(\omega,\sigma)$, for some $s$ and $k$, such that 
\begin{equation*}
\begin{split}
\dist(\pi,\omega)&=\sum_{i=1}^{k}\dist(\nu_{i},e),\\
\dist(\omega,\sigma)&=\sum_{i=k+1}^{s}\dist(\nu_{i},e),
\end{split}
\end{equation*}
and thus
\begin{equation*}
\begin{split}
\dist(\pi,\sigma)=\sum_{i=1}^{s}\dist(\nu_{i},e)\ge\min_{\left(\tau_{1},\cdots,\tau_{s'}\right)\in A(\pi,\sigma)}\sum_{i=1}^{s'}\dist(\tau_{i},e),
\end{split}
\end{equation*}
where the inequality follows from the fact that $\left(\nu_{1},\cdots,\nu_{s}\right)\in A(\pi,\sigma)$.
To complete the proof, note that by the triangle inequality, 
\[\dist(\pi,\sigma)\le\min_{\left(\tau_{1},\cdots,\tau_{s'}\right)\in A(\pi,\sigma)}\sum_{i=1}^{s'}\dist(\tau_{i},e).\]

\end{proof} 

\begin{defn} A distance $\dist_{\varphi}$ is termed a \emph{weighted Kendall distance }if there exists a nonnegative \emph{weight function} $\varphi$ over the set of adjacent transpositions such that \[ \dist_{\varphi}(\pi,\sigma)=\min_{\left(\tau_{1},\cdots,\tau_{s}\right)\in\adjset(\pi,\sigma)}\sum_{i=1}^{s}\varphi_{\tau_{i}}, \] where $\varphi_{\tau_i}$ is the weight assigned to transposition $\tau_i$ by $\varphi$.

The weight of a transform $\tau=\left(\tau_{1},\cdots,\tau_{s}\right)$ is denoted by $\wt\left(\tau\right)$ and is defined as \[ \wt(\tau)=\sum_{i=1}^{s}\varphi_{\tau_{i}}. \] Hence, $\dist_{\varphi}(\pi,\sigma)$ may be written as \[ \dist_{\varphi}(\pi,\sigma)=\min_{\tau\in\adjset(\pi,\sigma)}\wt(\tau). \]
\end{defn} 

Note that a weighted Kendall distance is completely determined by its weight function $\varphi$. 

\begin{thm} \label{thm:WKD}A distance $\dist$ satisfies Axioms II if and only if it is a weighted Kendall distance.\end{thm} 

\begin{proof} It follows immediately from Lemma \ref{lem:WKD} that a distance $\dist$ satisfying Axioms II is a weighted Kendall distance by letting \[ \varphi_{\theta}=\dist(\theta,e) \] for every transposition $\theta$ taken from the set of adjacent transpositions $\mathbb{A}_n$ in $\mathbb{S}_n$. 

The proof of the converse is omitted since it is easy to verify that a weighted Kendall distance satisfies Axioms II. \end{proof}

The weighted Kendall distance provides a natural solution for the top-vs-bottom issue. For instance, recall the example of ranking cities to live in, with
\begin{align*}
\pi = (&\text{Melbourne, Vienna, Vancouver, Toronto, Calgary,} \\
&\text{Adelaide, Sydney, Helsinki, Perth, Auckland}),\\
 \pi' = (&\text{Melbourne, Vienna, Vancouver, Calgary, Toronto,}\\
&\text{Adelaide, Sydney, Helsinki, Perth, Auckland}),\\
 \pi'' = (&\text{Vienna, Melbourne, Vancouver, Toronto, Calgary,}\\
 &\text{Adelaide, Sydney, Helsinki, Perth, Auckland}),\end{align*}
and choose the weight function 
\(\nwt{i}{i+1} = 0.9^{i-1}\)
for $i=1,2,\cdots,9$. Then, $\dist_\wtfn(\pi,\pi')=0.9^4=0.66<\dist_\wtfn(\pi,\pi'')=1$ as expected. In this case, we have chosen the weight function to be exponentially decreasing -- the choice of the weight function in general depends on the application. 

\subsection*{Computing the Weighted Kendall Distance for Monotonic Weight Functions}

Computing the weighted Kendall distance between two permutations for an arbitrary weight function is not as straightforward a task as computing the Kendall $\tau$ distance. However, in what follows, we show that for an important class of weight functions -- termed ``monotonic'' weight functions -- the weighted Kendall distance may be computed efficiently. 

\begin{defn} A weight function $\wtfn:\mathbb{A}_{n}\to\mathbb{R}^{+}$, where $\mathbb{A}_{n}$ as before denotes the set of adjacent transpositions in $\mathbb{S}_{n}$, is decreasing if $i>j$ implies that $\wtf{(i\,i+1)}\le\wtf{(j\,j+1)}.$ Increasing weight functions are defined similarly. A weight function is monotonic if it is increasing or decreasing. 
\end{defn} 

Monotonic weight functions are of importance in the top-vs-bottom model as they can be used to emphasize the significance of the top of the ranking by assigning higher weights to transpositions at the top of the list. An
example of a decreasing weight function is the exponential weight described in the previous subsection.

Suppose that $\tau=\left(\tau_{1},\cdots,\tau_{|\tau|}\right)$ of length $|\tau|$ transforms $\pi$ into $\sigma$. The transformation may be viewed as a sequence of moves of elements $i$, $i=1,\ldots,n,$ from position $\pi^{-1}(i)$ to position $\sigma^{-1}(i)$. 
Let the \emph{walk} followed by element $i$ while moved by the transform $\tau$ be denoted by $p^{i,\tau}=\left(p_{1}^{i,\tau},\cdots,p_{\left|p^{i,\tau}\right|+1}^{i,\tau}\right)$, where $\left|p^{i,\tau}\right|$ is the length of the walk $p^{i,\tau}$. 

For example, consider 
\begin{equation*}
\begin{split}
\pi &= (3,2,4,1),\\
\sigma & = (1,2,3,4),\\
\tau &= (\tau_1,\tau_2,\tau_3,\tau_4)\\
& = ((3\,4),(2\,3),(1\,2),(2\,3))
\end{split}
\end{equation*}
and note that $\sigma = \pi \tau_1 \tau_2 \tau_3 \tau_4$. We have
\begin{equation*}
\begin{split}
p^{1,\tau} &= (4,3,2,1),\\
p^{2,\tau} &= (2,3,2),\\
p^{3,\tau} &= (1,2,3),\\
p^{4,\tau} &= (3,4).
\end{split}
\end{equation*}

We first bound the lengths of the walks $p^{i,\tau},i\in[n].$ Let $I_{i}(\pi,\sigma)$ be the set consisting of elements $j\in[n]$ such that $\pi$ and $\sigma$ disagree on the pair $\{i,j\}$. 
In the transform $\tau$, all elements of $I_{i}(\pi,\sigma)$ must be swapped with $i$ by some $\tau_{k},k\in[|\tau|]$. Each such swap contributes length one to the total length of the walk 
$p^{i,\tau}$ and thus, $\left|p^{i,\tau}\right|\ge\left|I_{i}(\pi,\sigma)\right|$.

As before, let $\dist_{\wtfn}$ denote the weighted Kendall distance with weight function $\varphi$. Since for any $\tau\in A(\pi,\sigma)$, 
\[ \sum_{i=1}^{|\tau|}\varphi_{\tau_{i}}=\sum_{i=1}^{n}\frac{1}{2}\sum_{j=1}^{\left|p^{i,\tau}\right|}\nwt{p_{j}^{i,\tau}}{p_{j+1}^{i,\tau}}, \] 
we have
\[
\dist_{\wtfn}(\pi,\sigma)=\min_{\tau\in A(\pi,\sigma)} \sum_{i=1}^{n}\frac{1}{2}\sum_{j=1}^{\left|p^{i,\tau}\right|} \nwt{p_{j}^{i,\tau}}{p_{j+1}^{i,\tau}}. \] 
Thus,
\begin{equation} 
\dist_{\wtfn}(\pi,\sigma)\ge\sum_{i=1}^{n}\frac{1}{2} 
\min_{p^{i}\in P_{i}(\pi,\sigma)}
\sum_{j=1}^{\left|p^{i}\right|}\nwt{p_{j}^{i}}{p_{j+1}^{i}},
\label{eq:LB} \end{equation} 
where for each $i$, $P_{i}(\pi,\sigma)$ denotes the set of walks of length $\left|I_{i}(\pi,\sigma)\right|$, starting from $\pi^{-1}(i)$ and ending in $\sigma^{-1}(i)$. 
For convenience, let
\[p^{i,\star}(\pi,\sigma) = \arg\min_{p^{i}\in P_{i}(\pi,\sigma)}
\sum_{j=1}^{\left|p^{i}\right|}\nwt{p_{j}^{i}}{p_{j+1}^{i}}\]
be the minimum weight walk from $\pi^{-1}(i)$ to $\sigma^{-1}(i)$ with length $|I_i(\pi,\sigma)|$.

If clear from the context, we write $p^{i,\star}(\pi,\sigma)$ as $p^{i,\star}$.

We show next that for decreasing weight functions, the bound given in \eqref{eq:LB} is achievable and thus the value on the right-hand-side gives the weighted Kendall distance for this class of weight functions. 

Consider $\pi,\sigma\in\S_n$ and a decreasing weight function $\wtfn$. For each $i$, it follows that $p^{i,\star}(\pi,\sigma)$ extends to positions with largest possible indices, i.e., $p^{i,\star}=(\pi^{-1}(i),\allowbreak\cdots,\ell_{i}-1,\ell_{i},\ell_{i}-1,\cdots,\sigma^{-1}(i))$ where $\ell_{i}$ is the solution to the equation \[ \ell_{i}-\pi^{-1}(i)+\ell_{i}-\sigma^{-1}(i)=I_{i}(\pi,\sigma) \] and thus $\ell_{i}=\left(\pi^{-1}(i)+\sigma^{-1}(i)+I_{i}(\pi,\sigma)\right)/2.$

We show next that there exists a transform $\tau^{\star}$ with $p^{i,\tau^{\star}} =p^{i,\star},$ and so equality in (\ref{eq:LB}) can be achieved. The transform is described in Algorithm \ref{alg:FindTauMonotone}. The transform in question, $\tau^{\star}$, converts $\pi$ into $\sigma$ in $n$ rounds. In Algorithm~\ref{alg:FindTauMonotone}, the variable $r$ takes values $\sigma(1),\sigma(2),\cdots,\sigma(n)$, in that given order. 
For each value of $r$, $\tau^{\star}$ moves $r$ through a sequence of adjacent transpositions from its current position in $\pi_t$, $\pi_t^{-1}(r)$, to position $\sigma^{-1}(r)$. 

Fix $i\in[n]$. For values of $r$, used in Algorithm 1, such that $\sigma^{-1}(r)<\sigma^{-1}(i)$, $i$ is swapped with $r$ via an adjacent transposition if $\pi^{-1}(r)>\pi^{-1}(i)$. For $r=i$, $i$ is swapped with all elements $k$ such that $\pi^{-1}(k)<\pi^{-1}(i)$ and $\sigma^{-1}(i)<\sigma^{-1}(k)$. For $r$ such that $\sigma^{-1}(r)>\sigma^{-1}(i)$, $i$ is not swapped with other elements. Hence, $i$ is swapped precisely with elements of the set $I_i(\pi,\sigma)$ and thus, $|p^{i,\tau^\star}(\pi,\sigma)|=|I_i(\pi,\sigma)|$. Furthermore, it can be seen that, for each $i$, $p^{i,\tau^\star}(\pi,\sigma) = (\pi^{-1}(i),\cdots,\ell_{i}'-1,\ell_{i}',\ell_{i}'-1,\cdots,\sigma^{-1}(i)),$ for some $\ell_{i}'$. Since $|p^{i,\tau^\star}(\pi,\sigma)|=|I_i(\pi,\sigma)|$, $\ell_{i}'$ also satisfies the equation \[ \ell_{i}'-\pi^{-1}(i)+\ell_{i}'-\sigma^{-1}(i)=I_{i}(\pi,\sigma), \] implying that $\ell'_{i}=\ell_{i}$ and thus $p^{i,\tau^{\star}}=p^{i,\star}$. Consequently, one has the following result.

\begin{algorithm} \caption{FindTauMonotone\label{alg:FindTauMonotone}}

\renewcommand{\algorithmicrequire}{\textbf{Input:}} \renewcommand{\algorithmicensure}{\textbf{Output:}}

\begin{algorithmic}[1]

\REQUIRE $\pi,\sigma\in\mathbb{S}_{n}$

\vspace{.2mm}
\ENSURE $\tau^{\star}=\arg\min_{\tau\in A\left(\pi,\sigma\right)}\wt\left(\tau\right)$

\vspace{.2mm}
\STATE $\pi_0 \leftarrow \pi$

\vspace{.2mm}
\STATE $t\leftarrow0$

\vspace{.2mm}
\FOR{$r=\sigma(1),\sigma(2),\cdots,\sigma(n)$}

\vspace{.2mm}
\WHILE{$\pi_t^{-1}(r)>\sigma^{-1}(r)$}

\vspace{.2mm}
\STATE $\tau_{t+1}^{\star}\leftarrow\left(\pi_t^{-1}(r)-1\ \ \pi_t^{-1}(r)\right)$

\vspace{.2mm}
\STATE $\pi_{t+1}\leftarrow\pi_t\tau_{t+1}^{\star}$

\vspace{.2mm}
\STATE $t\leftarrow t+1$

\vspace{.2mm}
\ENDWHILE

\vspace{.2mm}
\ENDFOR

\end{algorithmic} \end{algorithm}

\begin{prop} For rankings $\pi,\sigma\in\mathbb{S}_{n}$, and a decreasing weighted Kendall weight function $\wtfn$, we have \[ \dist_{\wtfn}(\pi,\sigma)=\sum_{i=1}^{n}\frac{1}{2}\left(\sum_{j=\pi^{-1}(i)}^{\ell_{i}-1}\wtf{(j\,j+1)}+\sum_{j=\sigma^{-1}(i)}^{\ell_{i}-1}\wtf{(j\,j+1)}\right) \] where $\ell_{i}=\left(\pi^{-1}(i)+\sigma^{-1}(i)+I_{i}(\pi,\sigma)\right)/2$.\end{prop} 
Increasing weight functions may be analyzed similarly.

\begin{example} Consider the rankings $\pi=4312$ and $e=1234$ and a decreasing weight function $\wtfn$. We have $I_{i}(\pi,e)=2$ for $i=1,2$ and $I_{i}(\pi,e)=3$ for $i=3,4$. Furthermore, \begin{align*} \ell_{1} & =\frac{3+1+2}{2}=3, & p^{1,\star} & =(3,2,1),\\ \ell_{2} & =\frac{4+2+2}{2}=4, & p^{2,\star} & =(4,3,2),\\ \ell_{3} & =\frac{2+3+3}{2}=4, & p^{3,\star} & =(2,3,4,3),\\ \ell_{4} & =\frac{1+4+3}{2}=4, & p^{4,\star} & =(1,2,3,4). \end{align*} The minimum weight transformation is \[ \tau^{\star}=\left(\underbrace{(3\,2),(2\,1)}_{1},\underbrace{(4\,3),(3\,2)}_{2},\underbrace{(4\,3)}_{3}\right), \] where the numbers under the braces denote the value $r$ corresponding to the indicated transpositions. The distance between $\pi$ and $e$ is \[ \dist_{\wtfn}(\pi,e)=\nwt12+2\nwt23+2\nwt34. \]
\end{example} 


\begin{example} The bound given in \eqref{eq:LB} is not tight for general weight functions as seen in this example. Consider $\pi = (4,2,3,1)$, $\sigma = (1,2,3,4)$, and a weight function $\wtfn$ with $\nwt12=2,\nwt23=1$, and $\nwt34=2$. Note that the domain of $\wtfn$ is the set of adjacent transpositions. We have
\begin{equation*}
\begin{split}
p^{1,\star} &= (4,3,2,1),\\
p^{2,\star} &= (2,3,2),\\
p^{3,\star} &= (3,2,3),\\
p^{4,\star} &= (1,2,3,4).\\
\end{split}
\end{equation*}
Suppose that a transform $\tau$ exists such that $p^{i,\star} = p^{i,\tau},i = 1,2,3,4$. From $p^{i,\star}$, it follows that in $\tau$, transpositions $(1\,2)$ and $(3\,4)$ each appear once and $(2\,3)$ appears twice. It can be shown, by considering all possible re-orderings of $\{(1\,2),(1\,2),(2\,3),(2\,3),(2\,3)\}$ or by an application of \cite[Lemma 5]{farnoud2012sorting} that $\tau$ does not transform $\pi$ into $\sigma$. Hence, for this example, the lower bound \eqref{eq:LB} is not achievable. 
\end{example}


\subsection*{Weight Functions with Two Identical Non-zero Weights}
Another example of a weighted Kendall $\tau$ distance for which a closed form solution may be found is described below. 

For a pair of integers $a,b,1\le a<b<n$, define the weight function as: 
\begin{equation} \label{eqTwoIdWf}
\varphi_{(i\ i+1)}=\begin{cases}
1, & \quad i\in\left\{ a,b\right\} \\
0, & \quad\mbox{else},
\end{cases}
\end{equation}
i.e., a function which only penalizes moves involving candidates in positions $a$ and $b$. 

Such weight functions may be used in voting problems where one only penalizes moving a link from one page (say, top-ten page) to another page (say, ten-to-twenty page). In other words,
one only penalizes moving an item from a ``high-ranked" set of positions to ``average-rank" or ``low-rank" positions.

An algorithm for computing the weighted Kendall distance for this case is given in the Appendix.


\subsection*{Approximating the Weighted Kendall Distance for General Weight Functions}

The result of the previous subsection implies that at least for one class of weight functions that capture the importance of the top entries in a ranking, computing the weighted Kendall distance has time complexity $O(n^2)$. Hence, distance computation efficiency does not represent a bottleneck for the employment of this form of  the weighted Kendall distance. 

In what follows, we present a polynomial-time 2-approximation algorithm for computing the most general form of weighted Kendall distances, as well as two algorithms for computing this distance exactly. 
While the exact computation has super exponential time complexity, for a small number of candidates -- say, less than 10 -- the computation can be performed in reasonable time. 
A small number of candidates and a large number of voters are frequently encountered in social choice applications, but less frequently in computer science.

In order to approximate the weighted Kendall distance, $\dist_\wtfn(\pi,\sigma)$, we use the function $D_\wtfn(\pi,\sigma)$, defined as
\begin{equation*}
D_\wtfn(\pi,\sigma) = \sum_{i=1}^n w(\pi^{-1}(i):\sigma^{-1}(i)),
\end{equation*}
where 
\begin{equation*}
w(k:l) =  
\begin{cases}
\sum_{h=k}^{l-1}\nwt{h}{h+1}, & \text{if } k<l,\\
\sum_{h=l}^{k-1}\nwt{h}{h+1}, & \text{if } k>l,\\
0, & \text{if } k=l,
\end{cases}
\end{equation*}denotes the sum of the weights of adjacent transpositions $(k\,k+1),(k+1\,k+2),\cdots,(l-1\,l)$ if $k<l$, the sum of the weights of adjacent transpositions $(l\,l+1),(l+1\,l+2),\cdots,(k-1\,k)$ if $l<k$, and 0 if $k=l$.

The following proposition states lower and upper bounds for $\dist_\wtfn$ in terms of $D_\wtfn$. 
The propositions is useful in practice, since $D_\wtfn$ can be computed in time $O(n^2)$, and provides the desired 2-approximation. 

\begin{prop}\label{prop:wk2approximation}
For a weighted Kendall weight function $\wtfn$ and for permutations $\pi$ and $\sigma$, \[\frac12 D_\wtfn(\pi,\sigma) \le \dist_\wtfn(\pi,\sigma) \le D_\wtfn(\pi,\sigma).\]
\end{prop}
We omit the proof of the proposition, since it follows from a more general result stated in the next section, and only remark that the lower-bound presented above proposition is weaker than 
the lower-bound given by \eqref{eq:LB}.

Next, we discuss computing the \emph{exact} weighted Kendall distance via algorithms for finding minimum weight paths in graphs. As already pointed out, the Kendall $\tau$ and the weighted Kendall distance are graphic distances. In the latter case, we define a graph $G$ with vertex set indexed by $\Sn $ and an edge of weight $\varphi_{(i\ i+1)},i\in[n-1],$ between each pair of vertices $\pi$ and $\sigma$ for which there exists an $i$ such that $\pi=\sigma(i\ i+1)$. The numbers of vertices and edges of $G$ are $|V|=n!$ and $|E|=n!(n-1)/2$, respectively. Dijkstra's algorithm with Fibonacci heaps \cite{fredman:1987:dijkstra} for finding the minimum weight path in a graph provides the distances of all $\pi\in \Sn$ to the identity in time $O(|E|+|V|\log|V|)=O(n!\,n\log n)$. 

The complexity of the algorithm for finding the distance between $\pi\in\Sn$ and the identity may be actually shown to be $O(n(\kdist(\pi,e))!)$, which is significantly smaller than $\Omega(n!)$ for permutations at small Kendall $\tau$ distance. The minimum weight path algorithm is based on the following observation. For $\pi$ in $\Sn$, there exists a transform $\tau=(\tau_1,\tau_2,\cdots,\tau_m)$ of minimum weight that transforms $\pi$ into $e$, 
such that $m=\kdist(\pi,e)$. In other words, each transposition of $\tau$ eliminates one inversion when transforming $\pi$ into $e$. Hence, $\pi\tau_1$ has one less inversion than $\pi$. 
As a result,
\begin{align}\label{eqnalgrec}
\dist_\varphi(\pi,e) & =\min_{i:\pi(i)>\pi(i+1)}\left(\varphi_{(i\ i+1)}+\dist_\varphi(\pi(i\ i+1),e)\right)
\end{align}
Note that the minimum is taken over all positions $i$ for which $i$ and $i+1$ form an inversion, i.e., for which $\pi(i)>\pi(i+1)$. 
Suppose that computing the weighted Kendall distance between the identity and a permutation $\pi$, with $\kdist(\pi,e)=d$, can be performed in time $T_d$. From \eqref{eqnalgrec}, we have 
\begin{equation*}
T_d = a \, n + d\, T_{d-1},\quad \text{for }d\ge 2,
\end{equation*}
and $T_1 = a\, n,$ for some constant $a$. By letting $U_d = T_d/(a\, n\, d!)$, we obtain $U_d = U_{d-1}+\frac1{d!},d\ge2,$ and $U_1 = 1$. Hence, $U_d = \sum_{i=1}^d\frac1{i!}$. It can then be shown that $d! \, U_d=\lfloor d!(e-1)\rfloor,$ and thus $T_d = a\,n\lfloor d!\,(e-1)\rfloor=O(nd!)$.

The expression \eqref{eqnalgrec} can also be used to find the distances of all $\pi\in \Sn$ from the identity by first finding the distances of permutations $\pi\in \Sn$ with $\kdist(\pi,e)=1$, then finding the distances of permutations $\pi\in \Sn$ with $\kdist(\pi,e)=2$, and so on\footnote{Note that such an algorithm requires that the set of permutations at a given Kendall $\tau$ distance from the identity be known.}. Unfortunately, the average Kendall $\tau$ distance between a randomly chosen permutation and the identity is $\binom{n}{2}/2$ (see the derivation of this known and a related novel result regarding the weighted Kendall distance in the Appendix), which limits the applicability of this algorithm 
to uniformly and at random chosen votes limited.
 

\subsection*{Aggregation with Weighted Kendall Distances: Examples}

In order to explain the potential of the weighted Kendall distance in addressing the top-vs-bottom aggregation issue, in what follows, we present a number of examples that illustrate how the choice of the weight function influences the final form of the aggregate. We focus on \emph{decreasing weight functions} and compare our results to those obtained using the classical Kendall $\tau$ distance.

Throughout the remainder of the paper, we refer to a solution of the aggregation problem using the Kendall $\tau$ as a \emph{Kemeny aggregate}. All the aggregation results are obtained via exhaustive search since the examples are small and only used for illustrative purposes. Aggregation is, in general, a hard problem and we postpone the analysis of the complexity of  computing aggregate rankings, and aggregate approximation algorithms, until the next section.

\begin{example}\label{example5x5}
Consider the set of rankings listed in $\Sigma$, where each row represents a ranking (vote),
\begin{equation*}
\Sigma=\left(\begin{array}{ccccc}
4 & 1 & 2 & 5 & 3\\\hline
4 & 2 & 1 & 3 & 5\\\hline
1 & 4 & 5 & 2 & 3\\\hline
2 & 3 & 1 & 5 & 4\\\hline
5 & 3 & 1 & 2 & 4
\end{array}\right).
\end{equation*}
The Kemeny optimal solution for this set of rankings is $(1,4,2,5,3)$. Note that despite the fact that candidate $4$ was ranked twice at the top of the list -- more than any other candidate -- it is ranked only
second in the aggregate. This may be attributed to the fact that $4$ was ranked last by two voters.

Consider next the weight function $\wtfn^{(2/3)}$ with $\nwt{i}{i+1}^{(2/3)}=(2/3)^{i-1},i\in[4]$. The optimum aggregate ranking for this weight equals $(4,1,2,5,3).$ The optimum aggregate based on $\wtfn^{(2/3)}$ 
puts $4$ before $1$, similar to what a plurality vote would do\footnote{In \emph{plurality votes}, the candidate with the most first-place rankings is declared the winner.}. The reason behind this swap is that  $\wtfn^{(2/3)}$ emphasizes strong showings of a candidate and downplays its weak showings, 
since weak showings have a smaller effect on the distance as the weight function is decreasing. In other words, higher ranks are more important than lower ranks when determining the position of a candidate. 
\end{example}

\begin{example}\label{example5x4}
Consider the set of rankings listed in $\Sigma$,
\[\Sigma=\left(\begin{array}{cccc}
     1  &   4  &   2  &   3\\\hline
     1  &   4  &   3  &   2\\\hline
     2  &   3  &   1  &   4\\\hline
     4  &   2  &   3  &   1\\\hline
     3  &   2  &   4  &   1
\end{array}\right).\]
The Kemeny optimal solution is $(4,2,3,1)$. Note that although the majority of voters prefer $1$ to $4$, $1$ is ranked last and $4$ is ranked first. More precisely, 
we observe that according to the pairwise majority test, $1$ beats $4$ but loses to $2$ and $3$. On the other hand, $4$ is preferred to both $2$ and $3$ but, as mentioned before, loses to $1$. 
Problems like this do not arise due to a weakness of Kemeny's approach, but due to the inherent ``rational intractability'' of rank aggregation. 
As stated by Arrow\cite{arrow1963social}, for \emph{any} ``reasonable'' rank aggregation method, there exists a set of votes such that the aggregated ranking prefers one candidate to another while the majority of voters prefer the later to the former.

Let us now focus on a weighted Kendall distance with weight function $\nwt{i}{i+1}=(2/3)^{i-1}, i=1,2,3$. The optimal aggregate ranking for this distance equals $(1,4,2,3)$. Again, 
we see a candidate with both strong showings and weak showings, candidate $1$, beat a candidate with a rather average performance. Note that in this solution as well, 
there exist candidates for which the opinion of the majority is ignored: $1$ is placed before $2$ and $3$, while according to the pairwise majority opinion it loses to both.
\end{example}

\begin{example} Consider the set of rankings listed in $\Sigma$,
\[\Sigma=\left(
\begin{array}{ccccccc}
5 & 4 & 1 & 3 & 2 \\\hline
1 & 5 & 4 & 2 & 3 \\\hline
4 & 3 & 5 & 1 & 2 \\\hline
1 & 3 & 4 & 5 & 2 \\\hline
4 & 2 & 5 & 3 & 1 \\\hline
1 & 2 & 5 & 3 & 4 \\\hline
2 & 4 & 3 & 5 & 1
\end{array}\right).\]
With the weight function $\nwt{i}{i+1}=(2/3)^{i-1},i\in[4]$, the aggregate equals $(4, 1, 5, 2, 3)$. The winner is $4$, while the plurality rule winner is $1$ as it appears three times on the top. 
Next, we increase the rate of decay of the weight function and let $\nwt{i}{i+1}=(1/3)^{i-1},i\in[4]$. The solution now is $(1, 4, 2, 5, 3)$, and the winner is candidate $1$, the same as the plurality rule winner. 
This result is a consequence of the fact that the plurality winner is the aggregate based on the weighted Kendall distance with weight function $\wtfn^{(p)}$, 
\[\nwt{i}{i+1}^{(p)}=\begin{cases}
1, & \qquad i=1,\\
0, & \qquad\mbox{else}.
\end{cases}
 \]

The Kemeny aggregate is $(4, 5, 1, 2, 3)$.
\end{example}

A shortcoming of distance-based rank aggregation is that sometimes the solution is not unique, and that the possible solutions differ widely. The following example describes one such scenario.
\begin{example}\label{example7x5} 
Suppose that the votes are given by $\Sigma$,
\[
\Sigma=\left(\begin{array}{ccc}
1 & 2 & 3\\\hline
1 & 2 & 3\\\hline
3 & 2 & 1\\\hline
2 & 1 & 3
\end{array}\right).
\]
Here, the permutations $(1,2,3),(2,1,3)$ are the Kemeny optimal solutions, with cumulative distance 4 from $\Sigma$. When the Kemeny optimal solution is not unique, it may be possible to obtain a unique solution by using a non-uniform weight function. In this example, it can be shown that for any non-uniform weight function $\wtfn$ with $\nwt12>\nwt23$, the solution is unique, namely, $(1,2,3)$.

A similar situation occurs if the last vote is changed to $(2, 3, 1)$. In that case, the permutations $(1,2,3)$, $(2,1,3)$, and $(2,3,1)$ are the Kemeny optimal solutions with cumulative distance 5 from $\Sigma$. Again, for any non-uniform weight function $\wtfn$ with $\nwt12>\nwt23$ the solution is unique and equal to $(1,2,3)$.
\end{example}

To summarize, the above examples illustrate how a proper choice for the weighted Kendall distance insures that top ranks are emphasized and how one may over-rule a moderate number of low rankings using
a specialized distance formula. One may argue that certain generalizations of Borda's method, involving non-uniform gaps between ranking scores, may achieve similar goals. This is not the case, as will be illustrated in what follows. 

One major difference between generalized Borda and weighted Kendall distances is in the already mentioned \emph{majority criteria} ~\cite{hodge2005mathematics}, which states that the candidate ranked first by the majority of voters has to be ranked first in the aggregate\footnote{ Note that a candidate ranked first by the majority is a Condorcet candidate. It is desirable that an aggregation rule satisfy the majority criterion and indeed most do, including the Condorcet method, the plurality rule, the single transferable vote method, and the Coombs method.}. Borda's aggregate with an arbitrary score assignments does not have this property, 
while aggregates obtained via weighted Kendall distances with decreasing weights (not identically equal to zero) have this property.

We first show that the Borda method with a fixed, but otherwise arbitrary set of scores may not satisfy the majority criterion. 
We prove this claim for $n=3$. A similar argument can be used to establish this claim for $n>3$. 

Suppose, for simplicity, that the number $m$ of voters is odd and that, for each vote, a score $s_{i}$ is assigned to a candidate with rank $i$, $i=1,2,3.$ Here, we assume that 
$s_{1}>s_{2}>s_{3} \ge0$. Suppose also that $(m+1)/2$ of the votes equal $(a,b,c)$ and that $(m-1)/2$ of the votes equal $(b,c,a)$. 
Let the total Borda scores for candidates $a$ and $b$ be denoted by $S$ and $S'$, respectively. We have 
\begin{align*}
S & =\frac{m+1}{2}s_{1}+\frac{m-1}{2}s_{3},\\
S' & =\frac{m+1}{2}s_{2}+\frac{m-1}{2}s_{1},
\end{align*}
and thus $S-S'=s_{1}-m\left(\frac{s_{2}-s_{3}}{2}\right)-\frac{s_{2}+s_{3}}{2}$. If $m>\frac{2s_{1}-(s_{2}+s_{3})}{s_{2}-s_{3}}$, then $S-S'<0$ and Borda's method ranks $b$ higher than $a$. 
As a result, candidates $a$, ranked highest by more than half of the voters, is not ranked first according to Borda's rule. This is not the case with weighted Kendall distances, as shown below.

\begin{prop}
An aggregate ranking obtained using the weighted Kendall distance with a decreasing weight function not identically equal to zero satisfies the majority criterion.
\end{prop}
\begin{proof}
Suppose that the weight function is $\varphi,$ and let $w_i=\varphi_{(i\ i+1)}$.  Since $w$ is decreasing and not identically equal to zero, we have $w_{1}>0$. 
Let $a_1$ be a candidate that is ranked first by a majority of voters. Partition the set of votes into two sets, $C$ and $D$, where $C$ is the set of votes that rank $a_1$ first and $D$ is the set of votes that do not. Furthermore, denote the aggregate ranking by $\pi$.

Suppose that $a_{1}$ is not ranked first in $\pi$ and that $\pi$ is of the form \[(a_{2},\cdots,a_{i},a_{1},a_{i+1},\cdots,a_{n}), \] for some $i\ge2$. Let $\pi'=(a_1,a_2,\cdots,a_n)$. We show that
\[
\sum_{j=1}^{m}\dist_\varphi\left(\pi,\sigma_{j}\right)>
\sum_{j=1}^{m}\dist_\varphi\left(\pi',\sigma_{j}\right)
\]
which contradicts the optimality of $\pi$. Hence, $a_{1}$ must be ranked first in $\pi$.

For $\sigma\in C$, we have 
\begin{equation}\label{EqPi}
\dist_\varphi(\pi,\sigma) = \dist_\varphi(\pi,\pi')+\dist_\varphi(\pi',\sigma).
\end{equation}
To see the validity of this claim, note that if $\pi$ is to be transformed to $\sigma$ via Algorithm \ref{alg:FindTauMonotone}, it is first transformed to $\pi'$ by moving $a_1$ to the first position. 

For $\sigma\in D$, we have
\begin{equation}\label{EqPiPi}
\dist_\varphi(\pi',\sigma) \le \dist_\varphi(\pi',\pi)+\dist_\varphi(\pi,\sigma),
\end{equation}
which follows from the triangle inequality.

To complete the proof, we write
\begin{align*}
\sum_{j=1}^{m}\dist_{\varphi}(\pi,\sigma_{j}) & =\sum_{\sigma\in C}\dist_{\varphi}(\pi,\sigma)+\sum_{\sigma\in D}\dist_{\varphi}(\pi,\sigma)\\
 & \ge\sum_{\sigma\in C}\dist_{\varphi}(\pi',\sigma)+|C|\dist_{\varphi}(\pi,\pi')\\
&\quad+\sum_{\sigma\in C}\dist_{\varphi}(\pi',\sigma)-|D|\dist_{\varphi}(\pi,\pi')\\
 & =\sum_{j=1}^{m}\dist_{\varphi}(\pi',\sigma)+\left(|C|-|D|\right)\dist_{\varphi}(\pi,\pi')\\
 & >\sum_{j=1}^{m}\dist_{\varphi}(\pi',\sigma)
\end{align*}
where the first inequality follows from \eqref{EqPi} and \eqref{EqPiPi}, and the second inequality follows from the facts that $|C|>|D|$ and that $\dist_{\varphi}(\pi,\pi')\ge w_1>0$.

\ignore{ Let $A$ denote the set of rankings that placed $a_{1}$ ahead of $a_{i}$ and let $B$ be the set of rankings that placed some $a_{i},$ $i \neq 1$, ahead of $a_{1}$. Note that $C\subseteq A.$ For $\sigma\in A$, let $F_{\sigma}$ denote the set of candidates that appear after $a_{1}$ in $\sigma$ and for $\sigma\in B$, let $F_{\sigma}$ denote the set of elements that appear after $a_{i}$ in $\sigma$. Furthermore, let $F$ denote $\{a_{i+1},\cdots,a_n\}$, the set of candidates that appear after $a_1$ in $\pi$. We start by proving that \begin{equation} \sum_{j=1}^{m}\dist_\varphi(\pi',\sigma_{j})\le\sum_{j=1}^{m}\dist_\varphi(\pi,\sigma_{j}), \label{eq:onestep} \end{equation} where $\pi'=\pi(i-1\ i)=(a_{2},\cdots,a_{i-1},a_{1},a_{i},\cdots,a_{n}).$ From \eqref{eq:whichweight}, we have \begin{align} \sum_{j=1}^{m}\dist_\varphi(\pi,\sigma_{j})-\sum_{j=1}^{m}\dist_\varphi(\pi',\sigma_{j}) & =\sum_{\sigma\in A}w_{n-1-|F\cap F_{\sigma}|}-\sum_{\sigma\in B}w_{n-1-|F\cap F_{\sigma}|}\nonumber \\ & \ge\sum_{\sigma\in C}w_{n-1-|F\cap F_{\sigma}|}-\sum_{\sigma\in B}w_{n-1-|F\cap F_{\sigma}|}\nonumber \\ & =\sum_{\sigma\in C}w_{n-1-|F|}-\sum_{\sigma\in B}w_{n-1-|F\cap F_{\sigma}|}\nonumber \\ & \ge0\label{eq:12} \end{align} where the first inequality is a consequence of $C \subset A$ and the non-negativity of the weight function, while the last step follows from the facts that $w$ is decreasing, $|F\cap F_{\sigma}|\le\left|F\right|$, and $|C|>|B|$. This proves \eqref{eq:onestep}. Repeating the same argument yields \[ \sum_{j=1}^{m}\dist_\varphi\left(\left(a_{2},a_{1},a_{3},\cdots,a_{n}\right),\sigma_{j}\right)\le\sum_{j=1}^{m}\dist_\varphi\left((a_{2},\cdots,a_{i},a_{1},a_{i+1},\cdots,a_{n}),\sigma_{j}\right). \] To complete the proof of the proposition, we show that \[ \sum_{j=1}^{m}\dist_\varphi\left(\left(a_{1},a_{2},a_{3},\cdots,a_{n}\right),\sigma_{j}\right) <\sum_{j=1}^{m} \dist_\varphi\left(\left(a_{2},a_{1},a_{3},\cdots,a_{n}\right),\sigma_{j}\right). \] Let $\mu=\left(a_{2},a_{1},a_{3},\cdots,a_{n}\right)$ and $\mu'=\left(a_{1},a_{2},a_{3},\cdots,a_{n}\right)$. Additionally, let $D$ denote the set of votes that ranked $a_{2}$ higher than $a_{1}$, let $G$ denote the set $\{a_{3},a_{4},\cdots,a_{n}\}$, and let $G_{\sigma},\sigma\in D,$ denote the set of candidates ranked lower than $a_{2}$ in $\sigma$. Similar to \eqref{eq:12}, we have \begin{align*} \sum_{j=1}^{m}\dist_\varphi(\mu,\sigma_{j})- \sum_{j=1}^{m}\dist_\varphi(\mu',\sigma_{j}) & \ge\sum_{\sigma\in C}w_{n-1-|G|}-\sum_{\sigma\in D}w_{n-1-|G\cap G_{\sigma}|}\\ & =|C|w_{1}-\sum_{\sigma\in D}w_{n-1-|G\cap G_{\sigma}|}\\ & >0 \end{align*} where the last step follows from the facts that $|C|>|D|$, $w_{1}\ge w_{n-1-|G\cap G_{\sigma}|}$, and $w_{1}>0$. This completes the proof.}

\end{proof}

\section{Weighted Transposition Distance} \label{sec:weighted-cayley}

The definition of the Kendall $\tau$ distance and the weighted Kendall distance is based on transforming one permutation into another using \emph{adjacent} transpositions. If, instead, all transpositions are allowed --
including non-adjacent transpositions -- a more general distance measure, termed \emph{weighted transposition distance} is obtained. This distance measure, as will be demonstrated below, represents a generalization of the weighted Kendall distance suitable for addressing similarity issues among candidates. It is worth pointing out that the weighted transposition distance is not based on the axiomatic approach described in the previous section.

\begin{defn}
Consider a function $\varphi$ that assigns to each transpositions $\theta,$ a non-negative weight $\varphi_\theta$. The weight of a sequence of transpositions is defined as the sum of the weights of its transpositions. That is, the weight of the sequence $\tau=(\tau_1,\cdots,\tau_{|\tau|})$ of transpositions equals 
\[ \wt(\tau)=\sum_{i=1}^{|\tau|}\wtf{\tau_{i}}. \]
For simplicity, we also denote the weighted transposition distance between two permutations $\pi,\sigma\in\Sn$, with weight function $\varphi$,  by $\dist_\varphi$. This distance equals the minimum weight of a sequence $\tau = (\tau_1,\cdots,\tau_{|\tau|})$ of transpositions such that $\sigma=\pi \tau_1\cdots \tau_{|\tau|}$. As before, we refer to such a sequence of transpositions as a transform converting $\pi$ into $\sigma$ and let $A_T(\pi,\sigma)$ denote the set of transforms that convert $\pi$ into $\sigma$.
\end{defn}

 With this notation at hand, the weighted transposition distance between $\pi$ and $\sigma$ may be written as
\[ \dist_\varphi(\pi,\sigma)=\min_{\tau\in A_T(\pi,\sigma)}\wt(\tau). \]


\ignore{
We proceed by showing how Kemeny's axiomatic approach may be extended further to introduce a number of new distances metrics useful in different ranking scenarios.

The first distance applies when only certain subsets of transpositions are allowed -- for example, when only elements of a class may be reordered to obtain an aggregated ranking. \begin{defn} Consider a subset $G=\{g_{1},\cdots,g_{m}\}$ of $\mathbb{S}_{n}$ such that $g\in G$ implies that $g^{-1}\in G$. Rankings $\pi$ and $\sigma$ are $G-$adjacent if there exist $g\in G$ such that $\pi=\sigma g$.

A \emph{$G-$transformation} of $\pi$ into $\sigma$ is a vector $(g_{1},\cdots,g_{k}),k\in\mathbb{N}$, with $g_{i}\in G,i\in[k]$, such that $\sigma=\pi g_{1}g_{2}\cdots g_{k}$ where $k$ is the length of the $G-$transformation. The set of $G-$transformations of $\pi$ into $\sigma$ is denoted by $A_{G}(\pi,\sigma)$. A \emph{minimum $G-$transformation} is a $G-$transformation of minimum length.

Furthermore, $\omega$ is said to be \emph{$G-$between} $\pi$ and $\sigma$ if there exists a minimum $G-$transformation $(g_{1},\cdots,g_{k})$ of $\pi$ into $\sigma$ such that $\omega=\sigma g_{1}\cdots g_{j}$ for some $j\in[k]$. \end{defn}

\begin{defn} \label{def:uniform-distance}For a subset $G$ of $\mathbb{S}_{n}$, a function $\dist:\mathbb{S}_{n}\to[0,\infty]$ is said to be a \emph{uniform $G-$distance} if\end{defn} \begin{enumerate} \item $\dist$ is a metric. \item $\dist$ is left-invariant. \item For any $\pi,\sigma\in\mathbb{S}_{n}$, if $\omega$ is between $\pi$ and $\sigma$, then $\dist(\pi,\sigma)=\dist(\pi,\omega)+\dist(\omega,\sigma)$. \item The smallest positive distance is one. \end{enumerate} Note that Lemma \ref{lem:on-the-line} also applies to $G-$uniform distances.

As an example, consider $G$ to be the set \[ \mathbb{T}_{n}=\{(ab):a,b\in[n],a\neq b\} \] of all transpositions. \begin{thm} The uniform $\mathbb{T}_{n}-$distance exists and is unique. Namely, \[ \dist_{T}(\pi,\sigma)=\min\left\{ \left|\tau\right|:\tau\in A_{\mathbb{T}_{n}}(\pi,\sigma)\right\} , \] (commonly known as Cayley's distance) is the unique $\mathbb{T}_{n}-$distance.\end{thm} \begin{IEEEproof} It is easy to verify that $\dist_{T}$ satisfies Definition \ref{def:uniform-distance}. It thus suffices to show that for a uniform $\mathbb{T}_{n}-$distance $\dist$, we have \begin{equation} \dist(\pi,\sigma)=\min\left\{ \left|\tau\right|:\tau\in A_{\mathbb{T}_{n}}(\pi,\sigma)\right\} .\label{eq:min-trans} \end{equation} Let \[ \tau^{\star}=\arg\min\left\{ \left|\tau\right|:\tau\in A_{\mathbb{T}_{n}}(\pi,\sigma)\right\} . \] Definition \ref{def:uniform-distance}.3 implies that, for $\pi,\sigma\in\mathbb{S}_{n},$ \begin{equation} \dist(\pi,\sigma)=\sum_{i=1}^{\left|\tau^{\star}\right|}\dist(\tau_{i}^{\star},e).\label{eq:sum-trans} \end{equation} Below, we show all transpositions have the same distance from the identity, say $d$. Thus (\ref{eq:sum-trans}) becomes \[ \dist(\pi,\sigma)=\left|\tau^{\star}\right|d \] and since the minimum positive distance is one, we shall have \[ \dist(\pi,\sigma)=\left|\tau^{\star}\right| \] which proves (\ref{eq:min-trans}).

To prove the claim that all transpositions have the same distance from the identity, we show that for a uniform $\mathbb{T}_{n}-$distance $\dist$, \[ \dist\left((a\,b),e\right)=\dist\left((c\,d),e\right) \] for all transpositions $(a\,b)$ and $(c\,d)$. This is obvious for $\{a,b\}=\{c,d\}$. We prove the case that $a,b,c,$ and $d$ are distinct. A similar argument applies when $\{a,b\}$ and $\{c,d\}$ have one element in common. The argument parallels that of Lemma \ref{lem:two-ways}.

Let $\pi=(a\,b\,c\,d)$, $\omega=(a\,d)\pi$, $\eta=(c\,d)\omega$ and note that $e=(b\,c)\eta$. Since $\pi\dash\omega\dash\eta\dash e$, by Definition \ref{def:uniform-distance}.3 and left-invariance of $\dist$, we have \begin{equation} \dist(\pi,e)=\dist\left((a\,d),e\right)+\dist\left((c\,d),e\right)+\dist\left((b\,c),e\right).\label{eq:T-1} \end{equation} Similarly, let $\alpha=(b\,c)\pi$, $\beta=(a\,b)\alpha$, and note that $e=(a\,d)\beta$. This shows \begin{equation} \dist(\pi,e)=\dist\left((b\,c),e\right)+\dist\left((a\,b),e\right)+\dist\left((a\,d),e\right).\label{eq:T-2} \end{equation} Equating the right-hand-sids of (\ref{eq:T-1}) and (\ref{eq:T-2}) yields $\dist\left((a\,b),e\right)=\dist\left((c\,d),e\right)$. \end{IEEEproof}

\begin{defn} \label{def:weighted-distance}For a subset $G$ of $\mathbb{S}_{n}$, a function $\dist:\mathbb{S}_{n}\to[0,\infty]$ is said to be a \emph{weighted $G-$distance} if\end{defn} \begin{enumerate} \item $\dist$ is a pseudo-metric. \item $\dist$ is left-invariant. \item For any $\pi,\sigma\in\mathbb{S}_{n}$, if $\pi$ and $\sigma$ are not $G-$adjacent, then there exists an $\omega$ distinct from both, such that $\dist(\pi,\sigma)=\dist(\pi,\omega)+\dist(\omega,\sigma)$. \end{enumerate} It is immediate that a distance $\dist$ is a weighted $G-$distance if and only if \[ \dist(\pi,\sigma)=\min_{\tau\in A_{G}(\pi,\sigma)}\sum_{i=1}^{|\tau|}\dist\left(\tau_{i},e\right) \] for $\pi,\sigma\in\mathbb{S}_{n}$.

One way to arrive at a weighted $G-$distance is to assign ``weights'' or ``costs'' $\wtf g$ to elements $g$ of $G$ and define $\dist_{\varphi}(\pi,\sigma)$, the distance arising from the weight function $\varphi$, as the minimum weight of a $G-$transformation that transforms $\pi$ into $\sigma$. That is, for $\pi,\sigma\in\mathbb{S}_{n}$, 
\[ \dist_\varphi(\pi,\sigma)=\min_{\tau\in A_{G}(\pi,\sigma)}\wt(\tau) \] 
where \[ \wt(\tau)=\sum_{i=1}^{|\tau|}\wtf{\tau_{i}}. \] For example, assigning weights to transpositions leads to the weighted $\mathbb{T}_{n}-$distance (or weighted transposition distance).}


The Kendall $\tau$ distance and the weighted Kendall distance may be viewed as special cases of the weighted transposition distance: to obtain 
the Kendall $\tau$ distance, let \[ \wtf{\theta}=\begin{cases} 1, & \qquad\theta=(i\,i+1),i=1,2,\cdots,n-1\\ \infty, & \qquad\mbox{else}, \end{cases} \] and to obtain the weighted Kendall distance, let \[ \wtf{\theta}=\begin{cases} w_{i}, & \qquad\theta=(i\,i+1),i=1,2,\cdots,n-1\\ \infty, & \qquad\mbox{else}, \end{cases} \] for a non-negative weight function $w$.

When applied to the inverse of rankings, the weighted transposition distance can be successfully used to model similarities of objects in rankings. In such a setting, permutations that differ by a transposition of two similar 
items are at a smaller distance than permutations that differ by a transposition of two dissimilar items, as demonstrated in the next subsection.

\ignore{\begin{rem} Note that similar to the generalization of Kemeny's axioms one may arrive at a generalization of Borda's score-based rule. A step in this direction was proposed by Young \cite{Young:1975fk}, who showed that a set of axioms leads to a generalization of Borda's rule wherein the $k$th preference of each ranking receives a score $s_{k}$, not necessarily equal to $k$. This generalization of Borda's rule may also be used to address the problem of top versus bottom in rankings. In particular, one may assign Borda scores $s_{k}$ to the $k$th preference with \begin{align*} s_{k}=\sum_{l=1}^{k-1}\phi_{l} & , \end{align*} where $\phi_{k}$ is decreasing in $l$. For example, swapping two elements at the top of the ranking of a given voter changes the scores of each of the two corresponding objects by $\phi_{1}$ while a similar swap at the bottom of the ranking, changes the scores by $\phi_{n-1}$. Since $\phi_{1}\ge\phi_{n-1}$, changes to the top of the list, in general, have a more significant affect on the aggregate ranking. \end{rem}}

\subsection{Weighted Transposition Distance as Similarity Distance: Examples}


We illustrate the concept of distance measures taking into account similarities via the following example, already mentioned in the Motivation section. 
Suppose that four cities: Melbourne, Sydney, Helsinki, and Vienna are ranked based on certain criteria as 
\[\pi = (\text{Helsinki, Sydney, Vienna, Melbourne}),\]
 and according to another set of criteria as 
\[\sigma = (\text{Melbourne, Vienna, Helsinki, Sydney}).\]

The distance between $\pi$ and $\sigma$ is defined as follows. We assign weights to swapping the positions of cities in the rankings, e.g., suppose that the weight of swapping cities in the same country is $1$, on the same continent $2$, and $3$ otherwise. The \emph{similarity distance} between $\pi$ and $\sigma$ is the minimum weight of a sequence of transpositions that convert $\pi$ into $\sigma$. By inspection, or by methods discussed in the next subsection, one can see that the similarity distance between $\pi$ and $\sigma$ equals 6. One of the sequences of transpositions of weight 6 is as follows: first swap Helsinki and Sydney with weight 3, then swap Melbourne and Sydney with weight 1, and finally swap Vienna and Helsinki with weight 2.

To express the similarity distance formally, we write the rankings as permutations, representing Melbourne by $1$, Sydney by $2$, Vienna by $3$, and Helsinki by $4$. 
This is equivalent to assuming that the identity ranking is \[e=(\text{Melbourne, Sydney, Vienna, Helsinki}).\]
We then have $\pi = (4,2,3,1)$ and $\sigma = (1,4,2,3)$. The weight function $\wtfn$ equals 
\begin{eqnarray*}
\nwt12= 1, & \nwt13= 3, & \nwt14 = 3 \\
\nwt23= 3, & \nwt24= 3, & \nwt34 = 2.
\end{eqnarray*}
It should be clear from the context that the indices in the weight function refer to the candidates, and not to their positions.
%
%

\begin{example}
Consider the votes listed in $\Sigma$ below, 
\[ \Sigma=\left(\begin{array}{cccc}
1 & 2 & 3 & 4\\\hline
3 & 2 & 1 & 4\\\hline
4 & 1 & 3 & 2
\end{array}\right). \]
Suppose that even numbers and odd numbers represent different types of candidates in a way that the following weight function is appropriate
\[
\nwt{i}j =\begin{cases}
1, & \qquad \text{if}\, i,j\mbox{ are both odd or both even},\\
2, & \qquad\mbox{else.}
\end{cases}\]
Note that the votes are ``diverse'' in the sense that they alternate between odd and even numbers. On the other hand, the Kemeny aggregate is $(1, 3, 2, 4),$ which puts all odd numbers ahead of all even numbers. 
Aggregation using the similarity distance described above yields $(1,2,3,4),$ a solution which may be considered ``diverse'' since the even and odd numbers alternate in the solution. The reason behind this result is 
that the Kemeny optimal solution is oblivious to the identity of the candidates and their (dis)similarities, while aggregation based on similarity distances take such information into account.
\end{example}

\ignore{
\subsubsection*{Similarity Distance and Assignment Aggregation}
Next, we provide another application of similarity distance in a context that is close related to rank aggregation. Consider a committee with $m$ members that is tasked with filling $n$ jobs with $n$ candidates. Suppose that each committee member provides a full assignment of candidates to jobs. The task is to aggregate the assignments given by individual committee members into one assignment. If a measure of similarity between the candidates is available, one can use similarity distance to aggregate the assignment.  Let $\wtfn$ be a weight function reflecting similarities between the candidates, i.e., more similar a pair of candidates, less the weight of swapping them. Then we choose 
\[\arg \min_{\pi\in\S _n} \sum \dist'_\wtfn (\pi,\sigma_i),\]
where $\sigma_i$ are the assignments given by the committee members, as the aggregate assignment.}

\begin{example}
Consider the votes listed in $\Sigma$ below, 
\[ \Sigma=\left(\begin{array}{cccccc}
1 & 2 & 3 & 4 & 5 & 6\\\hline
1 & 2 & 3 & 4 & 5 & 6\\\hline
3 & 6 & 5 & 2 & 1 & 4\\\hline
3 & 6 & 5 & 2 & 1 & 4\\\hline
5 & 4 & 1 & 6 & 3 &2.
\end{array}\right). \]
Suppose that the weight function is the same as the one used in the previous example. In this case, neither the Kemeny aggregates nor the weighted transposition distance aggregates are unique.
More precisely, Kendall $\tau$ gives four solutions:
\[ \left(\begin{array}{cccccc}
3 & 5 & 1 & 6 & 2 & 4\\\hline
3 & 5 & 1 & 2 & 4 & 6\\\hline
1 & 3 & 5 & 2 & 4 &6.\\ \hline
1 & 3 & 5 & 6 & 2 &4.
\end{array}\right), \]
while there exist nine optimal aggregates under the weighted transposition distance of the previous example, of total distance $10$:
\[ \left(\begin{array}{cccccc}
5 & 6 & 3 & 4 & 1 & 2\\\hline
5 & 4 & 3 & 2 & 1 & 6\\\hline
5 & 2 & 3 & 6 & 1 & 4\\\hline
3 & 4 & 1 & 2 & 5 & 6\\\hline
3 & 6 & 1 & 4 & 5 & 2\\\hline
3 & 2 & 1 & 6 & 5 & 4\\\hline
1 & 4 & 5 & 2 & 3 & 6 \\ \hline
1 & 2 & 5 & 6 & 3 & 4 \\ \hline
1 & 6 & 5 & 4 & 3 & 2.
\end{array}\right). \]

Note that \emph{none of} the Kemeny optimal aggregates have good diversity properties: the top half of the rankings consists exclusively of odd numbers. On the other hand,
the optimal weighted transposition rankings \emph{all contain} exactly one even element among the top-three candidates. Such diversity properties are hard to prove theoretically.

\end{example}

\subsection{Computing the Weighted Transposition Distance}

In this subsection, we describe how to compute or approximate the weighted transposition distance $\dist_{\varphi},$ given the weight function $\varphi$. An in-depth analysis of a special class of weight functions
and their corresponding transposition distance may be found in the authors' recent work~\cite{farnoud2012sorting}. 

We find the following definitions useful in our subsequent derivations. For a given weight function $\wtfn$, we let $\mathcal K_\wtfn$ denote a complete undirected weighted graph with vertex set $[n]$, 
where the weight of each edge $(i,j)$ equals the weight of the transposition $(i\,j)$, $\wtf{(i\,j)}$. For a subgraph $H$ of $\mathcal{K}_{\varphi},$ with edge set $E_{H},$ we define the weight of $H$ as 
\[ \wt(H)=\sum_{(i,j)\in E_{H}}\wtf{(i\,j)}, \] 
that is, the sum of the weights of edges of $H$. For $\pi,\sigma\in\S_n$,  we define $D_\wtfn(\pi,\sigma)$ as\footnote{Note that this definition is consistent with the definition of a specialization of this function, given in Proposition 16.}
\begin{equation*}
D_\wtfn(\pi,\sigma) = \sum_{i=1}^n \wt\left(p^*_\wtfn(\pi^{-1}(i),\sigma^{-1}(i))\right),
\end{equation*}
where $p^*_\wtfn(a,b)$ denotes the minimum weight path from $a$ to $b$ in $\mathcal{K}_{\wtfn}$.

It is easy to verify that $D_\wtfn$ is a pseudo-metric and that it is left-invariant,
\begin{equation*}
D_\wtfn(\eta\pi,\eta\sigma) = D_\wtfn(\pi,\sigma),\qquad \pi,\sigma,\eta\in\Sn.
\end{equation*}


A weight function $\wtfn$ is a \emph{metric weight function} if it satisfies the triangle inequality in the sense that
\begin{equation}\label{eqwftri}
\nwt{a}b \le \nwt{a}c + \nwt{b}c,\qquad a,b,c\in[n].
\end{equation}

\begin{lem}\label{lemTranspositionUB} For a non-negative weight function $\wtfn$ and a transposition $(a\,b)\in\Sn$, 
\begin{equation*}
\dist_\wtfn((a\,b),e) \le 2 \wt (p^*_\wtfn(a,b)).
\end{equation*} 
If $\wtfn$ is a metric weight function, the bound may be improved to
\begin{equation*}
\dist_\wtfn((a\,b),e) \le \wt (p^*_\wtfn(a,b)).
\end{equation*}
\end{lem}

\begin{proof}
Consider a path $p=(v_{0}=a,v_{1},\cdots,v_{|p|}=b)$ from $a$ to $b$ in $\mathcal{K}_{\varphi}$. We have
\begin{align*} 
(a\,b)=&\left(v_{0}\,v_{1}\right) \left(v_{1}\,v_{2}\right) \cdots \left(v_{|p|-2}\,v_{|p|-1}\right) \\
&\quad\left(v_{|p|-1}\,v_{|p|}\right) \left(v_{|p|-2}\,v_{|p|-1}\right) \cdots \left(v_{1}\,v_{2}\right) \left(v_{0}\,v_{1}\right).
\end{align*} 
From the left-invariance of $\dist_\varphi$,
\begin{align*} 
\dist_\varphi((a\,b),e)
& =2\sum_{i=1}^{|p|-1}\varphi_{(v_{i-1}\,v_{i})}-\varphi_{(v_{|p|-1}\,v_{|p|})}\\
& =2\wt(p)-\varphi_{(v_{|p|-1}\,v_{|p|})}\\ 
& \le2\wt(p). 
\end{align*}
Since $p$ is an arbitrary path from $a$ to $b$ in $\mathcal{K}_{\varphi}$, we have 
\begin{equation}
\dist_{\varphi}\left((a\,b),e\right) \le\wt(\tau)\nonumber \le2\wt(p^{*}_\wtfn(a,b)),\label{eq:dist-trans}
\end{equation} 
and this proves the first claim.

Now, assume that $\wtfn$ is a metric weight function and consider the path $p=(v_{0},v_{1},\cdots,v_{|p|})$ from $v_0=a$ to $v_{|p|}=b$. From \eqref{eqwftri},
\begin{equation*}
\begin{split}
\nwt{a}b &= \nwt{v_0}{v_{|p|}} \le \nwt{v_0}{v_1} + \nwt{v_1}{v_{|p|}}\\
& \le \nwt{v_0}{v_1} + \nwt{v_1}{v_{2}}+\nwt{v_2}{v_{|p|}}\\
& \le \cdots\\
& \le \sum_{i=1}^{|p|-1}\nwt{v_i}{v_{i+1}}\\
& = \wt(p).
\end{split}
\end{equation*}
Since $p$ is arbitrary, we have \[\dist_\wtfn((a\,b),e)\le\nwt{a}b\le\wt(p^*_\wtfn(a,b)).\] 
This completes the proof of the lemma.
\end{proof}

While Lemma~\ref{lemTranspositionUB} suffices to prove all our subsequent results, we remark that one may prove a slightly stronger result, presented in our companion paper~\cite{farnoud2012sorting},
 \[\dist_\varphi((a\ b),e)=\min_{p=(v_0=a,v_1,\cdots,v_{|p|}=b)} \left(2\wt(p)-\max_{0\le i<|p|} \varphi_{(v_i\,v_{i+1})}\right).\]
The proof is based on significantly more involved techniques that are beyond the scope of this paper.

\begin{lem}\label{lemPermutationUB} For a weight function $\wtfn$ and  for $\pi,\sigma\in\S_n$,
\begin{equation*}
\dist_\wtfn(\pi,\sigma) \le 2 D_\wtfn(\pi,\sigma).
\end{equation*} 
If $\wtfn$ is a metric weight function, the bound may be improved to
\begin{equation*}
\dist_\wtfn(\pi,\sigma) \le D_\wtfn(\pi,\sigma).
\end{equation*} \end{lem} 
\begin{proof}
To prove the first claim, it suffices to show that $\dist_\wtfn(\pi,e) \le 2 D_\wtfn(\pi,e)$ since both $\dist_\wtfn$ and $D_\wtfn$ are left-invariant.

Let $\{c_1,c_2,\cdots,c_k\}$ be the cycle decomposition of $\pi$. We have, from the triangle inequality and the left-invariance property of $\dist_\varphi$, that
\[\dist_\wtfn(\pi,e) \le \sum_{i=1}^k \dist_\wtfn(c_i,e),\] and, from the definition of $D_\wtfn$, that
\[D_\wtfn(\pi,e) = \sum_{i=1}^k D_\wtfn(c_i,e).\] Hence, we only need to prove that
\begin{equation}
\dist_\wtfn(c,e) \le 2 D_\wtfn(c,e)
\end{equation}
for a single cycle $c=(a_1\,a_2\,\cdots\, a_{|c|}),$ where $|c|$ is the length of $c$.

Since $c$ may be written as
\begin{equation*}
c = (a_1\,a_2)(a_2\,a_3)\cdots(a_{|c|-1}\,a_{|c|}),
\end{equation*}
we have
\begin{equation*}
\begin{split}
\dist_\wtfn(c,e) &\le \sum_{i=1}^{|c|-1}\nwt{a_i}{a_{i+1}}\\
&\stackrel{\rm (a)}{\le} \sum_{i=1}^{|c|-1}2 \wt(p^*_\wtfn(a_i,a_{i+1}))\\
&\le \sum_{i=1}^{|c|}2 \wt(p^*_\wtfn(a_i,c(a_{i})))\\
&\le 2 D_\wtfn(c,e)
\end{split}
\end{equation*}
where ${\rm (a)}$ follows from Lemma~\ref{lemTranspositionUB}.

The proof of the second claim is similar.
\end{proof}

The next lemma provides a lower bound for $\dist_\varphi$ in terms of $D_\wtfn$.

\begin{lem}\label{lem:LB}
For $\pi,\sigma\in\S_n$, 
\[\dist_\wtfn(\pi,\sigma)\ge\frac12 D_\wtfn(\pi,\sigma).\]
\end{lem}
\begin{proof}
Since $\dist_\wtfn$ and $D_\wtfn$ are both left-invariant, it suffices to show that
\[\dist_\wtfn(\pi,e)\ge\frac12 D_\wtfn(\pi,e).\]
Let $\left(\tau_{1},\cdots,\tau_{l}\right),$ with $\tau_{j}=(a_{j}\,b_{j}),$ be a minimum weight transform of $\pi$ into $e$, so that $\dist_\varphi(\pi,e)=\sum_{i=1}^l\varphi_{(a_j\,b_j)}$. Furthermore, define $\pi_{j}=\pi\tau_{1}\cdots\tau_{j}$, $0\le j\le l$. Then, 
\begin{align} 
D_\wtfn\left(\pi_{j-1},e\right)-D_\wtfn\left(\pi_{j},e\right) &
\le2\wt\left(p^{*}_\wtfn\left(a_{j},b_{j}\right)\right)\nonumber \\ &
\le2\varphi_{(a_j\,b_j)},
\label{eq:Dj} 
\end{align} 
where the first inequality follows from considering the maximum possible decrease of the value of $D_\wtfn$ induced by one transposition, while the second inequality follows from the definition of $p^*_\varphi$. By summing up the terms in (\ref{eq:Dj}) over $0\le j\le l$, and thus obtaining a telescoping inequality of the form $D_\wtfn(\pi,e) \le2\sum_{i=1}^l\varphi_{(a_j\,b_j)} =2\dist_{\varphi}(\pi,e),$ we arrive at the desired result.
\end{proof}

From the previous two lemmas, we have the following theorem.
\begin{thm}\label{thm:UBLB}
For $\pi,\sigma\in\S_n$ and an arbitrary non-negative weight function $\wtfn$, we have
\[\frac12 D_\wtfn(\pi,\sigma)
\le\dist_\wtfn(\pi,\sigma)\le 2 D_\wtfn(\pi,\sigma). \]
In addition, if $\wtfn$ is a metric weight function, then
\[\frac12 D_\wtfn(\pi,\sigma)
\le\dist_\wtfn(\pi,\sigma)\le D_\wtfn(\pi,\sigma). \]
\end{thm}

For special classes of the weight function $\wtfn$, the bounds in Theorem~\ref{thm:UBLB} may be improved further, as described in the next subsection.

\subsection{Computing the Transposition Distance for Metric-Tree Weights}
\begin{figure} 
\begin{center}
\subfloat[]{\includegraphics[width=3.45in]{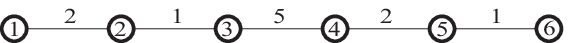}}

\subfloat[]{\includegraphics[width=3.45in]{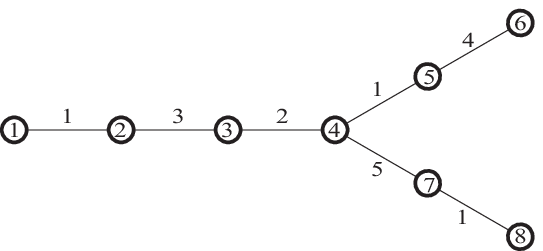}}

\end{center}
\caption{A defining path (a), which may correspond to a metric-path weight function or an extended-path weight function, and a defining tree (b), which may correspond to a metric-tree weight function or an extended-tree weight function.} 
\label{FigProjdefpath}

\end{figure}

We start with the following definitions.
\begin{defn} 
A weight function $\wtfn$ is a \emph{metric-tree weight function} if there exists a weighted tree $\Theta$ over the vertex set $[n]$ such that for distinct $a,b\in[n]$, $\nwt{a}b$ is the sum of the weights of the edges on the unique path from $a$ to $b$ in $\Theta$. If $\Theta$ is a path, i.e., if $\Theta$ is a linear graph, then $\wtfn$ is called a \emph{metric-path weight function}. 

Furthermore, a weight function $\wtfn'$ is an \emph{extended-tree weight function}  if there exists a weighted tree $\Theta$ over the vertex set $[n]$ such that for distinct $a,b\in[n]$, $\nwt{a}b'$ equals the the weight of the edge $(a,b)$ whenever $a$ and $b$ are adjacent, and $\nwt{a}b'=\infty$ otherwise. If $\Theta$ is a path, then $\wtfn'$ is called an \emph{extended-path} weight function.
\end{defn} 

Note that the Kendall weight function, defined in the previous section, is an extended path weight function.

The tree or path corresponding to a weight function in the above definitions is termed the \emph{defining tree or path} of the weight function. An example is given in \figurename~\ref{FigProjdefpath}, where the numbers 
indexing the edges denote their weights. 

For a metric-tree weight function $\varphi$ with defining tree $\Theta$, and for $a,b\in[n]$, the weight of the path $p^*_\varphi(a,b)$ equals the weight of the \emph{unique} 
path from $a$ to $b$ in $\Theta$. This weight, in turn, equals $\varphi_{(a\,b)}$. As a result, for metric-tree weights, $p^*_\varphi(a,b)$ equals the weight of the path from $a$ to $b$ in $\Theta$. 

Furthermore, from Lemma~\ref{lem:LB}, we have $\dist_\varphi((a\ b),e)\ge \frac12D_\varphi((a\ b),e)=\wt(p^*_\varphi(a,b))=\varphi_{(a\,b)}$. Since we also have $\dist_\varphi((a\ b),e)\le\varphi_{(a\,b)}$, it follows that 
\begin{equation}\label{eqMetricTrans}
\dist_\varphi((a\ b),e)=\varphi_{(a\,b)}.
\end{equation}

The next lemma shows that the \emph{exact distance} for metric-path weight functions can be computed in polynomial time. 
\begin{figure}
\begin{center}
\subfloat[]{\includegraphics[width=3.45in]{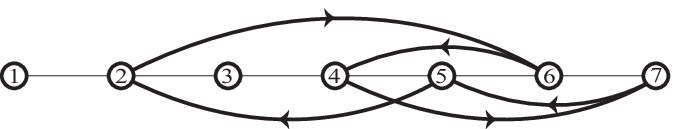}}

\subfloat[]{\includegraphics[width=3.45in]{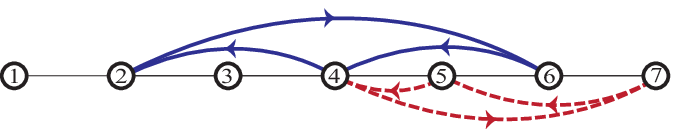}}

\end{center}
\caption{The cycle $(2\ 6\ 4\ 7\ 5)$ in Figure (a) is decomposed into two cycles, $(2\ 6\ 4)$ and $(4\ 7\ 5)$, depicted in Figure (b). Note that $ (2\ 6\ 4\ 7\ 5) = (2\ 6\ 4)(4\ 7\ 5)$.}
\label{FigDecomposePath}
\end{figure}

\begin{lem}\label{lem:metric-path}
For a metric-path weight function $\wtfn$ and for $\pi,\sigma\in\S_n$,
\[\dist_\wtfn(\pi,\sigma) = \frac12 D_\wtfn(\pi,\sigma).\]
\end{lem}

\begin{proof}
From Lemma~\ref{lem:LB}, we have that $\dist_\wtfn(\pi,\sigma) \ge \frac12 D_\wtfn(\pi,\sigma)$. It remains to show that $ \dist_\wtfn(\pi,\sigma) \le \frac12 D_\wtfn(\pi,\sigma)$. Since $\dist_\wtfn$ and $D_\wtfn$ are both left-invariant, it suffices to prove that $\dist_\wtfn(\pi,e) \le \frac12 D_\wtfn(\pi,e)$. 

Let $\{c_1,c_2,\cdots,c_k\}$ be the cycle decomposition of $\pi$. Similar to the proof of Lemma~\ref{lemPermutationUB}, it suffices to show that
\begin{equation}\dist_\wtfn(c,e) \le \frac12D_\wtfn(c,e)\label{eq:induchyp}\end{equation}
for any cycle $c=(a_1\,a_2\,\cdots\, a_{|c|})$. 

The proof is by induction. For $|c|=2$, \eqref{eq:induchyp} holds since, from \eqref{eqMetricTrans}, we have 
\[\dist_\wtfn((a_1\,a_2),e) = \varphi_{(a\,b)}=\wt\left(p^*_\wtfn(a_1,a_2)\right) = \frac12D_\wtfn((a_1\,a_2),e).\]
Assume that \eqref{eq:induchyp} holds for $2\le|c|<l$. We show that it also holds for $|c|=l$. We use Figure~\ref{FigDecomposePath} for illustrative purposes. 
In all figures in this section, undirected edges describe the defining tree, while directed edges describe the cycle at hand.

Without loss of generality, assume that the defining path of $\varphi$, $\Theta$, equals $(1,2,\cdots,n)$. Furthermore, assume that $a_1 = \min\{i:i\in c\};$ if this were not the case, we could 
rewrite $c$ by cyclically shifting its elements. 
Let $a_t=\min\{i:i\in c,i\neq a_1\}$ be the ``closest'' element to $a_1$ in $\Theta$ (that is, the closest element to $a_1$ in the cycle $c$).  
For example, in Figure~\ref{FigDecomposePath}, one has $c=(2\,6\,4\,7\,5)$, $a_1 = 2$ and $a_t = 4$. We have 
\begin{equation*}
\begin{split}
c & = (a_1\,a_2\,\cdots\,a_t\,\cdots\,a_l)\\
& = (a_1\ a_2\ \cdots\ a_t)(a_t\ a_{t+1}\ \cdots\ a_l)
\end{split}
\end{equation*}
and thus
\begin{equation*}
\begin{split}
\dist_\wtfn(c,e) &\le \dist_\wtfn((a_1\ a_2\ \cdots\ a_t),e)
+\dist_\wtfn((a_t\ a_{t+1}\ \cdots\ a_l),e)\\
&\le \frac12\sum_{i=1}^{t-1}\wt\left(p^*_\wtfn(a_i,a_{i+1})\right)+\wt\left(p^*_\wtfn(a_t,a_1)\right)\\
&\quad+\frac12\sum_{i=t}^{l-1}\wt\left(p^*_\wtfn(a_i,a_{i+1})\right)+\wt\left(p^*_\wtfn(a_l,a_t)\right)\\
& = \frac12\sum_{i=1}^l \wt\left(p^*_\wtfn(a_i,c(a_i))\right)\\
& = \frac12D_\wtfn(c,e).
\end{split}
\end{equation*}
where the second inequality follows from the induction hypothesis, while the first equality follows from the fact that $\wt\left(p^*_\wtfn(a_t,a_1)\right)+\wt\left(p^*_\wtfn(a_l,a_t)\right) =\wt\left(p^*_\wtfn(a_l,a_1)\right)$.
\end{proof}

The approach described in the proof of Lemma~\ref{lem:metric-path} can also be applied to the problem of finding the weighted transposition distance when the weight function is a metric-tree weight function and 
\emph{each of the cycles of the permutation} consist of elements that lie on some path in the defining tree. 
An example of such a permutation and such a weight function is shown in Figure~\ref{FigManyCyclesOnPath}. Note that in this example, a cycle consisting of elements $3,5,7$ would not correspond to a path.

In such a case, for each cycle $c$ of $\pi$ we can use the path in the defining tree that contains the elements of $c$ to show that 
\begin{equation}\label{EqExD}
\dist_\varphi(c,e)=\frac12D_\wtfn(c,e).
\end{equation} For example the cycle $(1\ 4\ 6)$ lies on the path $(1,2,3,4,5,6)$ and the cycle $(5\ 8)$ lies on the path $(5,4,7,8)$. 
Since \eqref{EqExD} holds for each cycle $c$ of $\pi$, we have \[\dist_\varphi(\pi,e)=\frac12D_\wtfn(\pi,e).\]
\begin{figure}
\begin{center}
\includegraphics[width=3in]{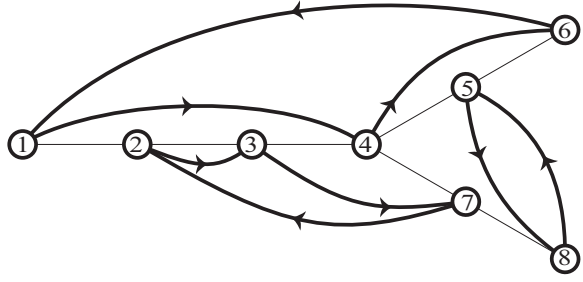}
\caption{If each of the cycles of a permutation lie on a path, the method of Lemma~\ref{lem:metric-path} can be used to find the weighted transposition distance.} 
\label{FigManyCyclesOnPath}
\end{center}
\end{figure}

A similar scenario in which essentially the same argument as that of the proof of Lemma~\ref{lem:metric-path} can be used is as follows: 
the defining tree has one vertex with degree three and no vertices with degree larger than three (i.e., a tree with a Y shape), and for each cycle of $\pi$, there are two branches of the tree that do not contain two consecutive elements of $c$. 
It can then be shown that each such cycle can be decomposed into cycles that lie on paths in the defining tree, reducing the problem to the previously described one. An example is shown in Figure~\ref{FigDecomposeTree}.
\begin{figure}
\begin{center}
\subfloat[]{\includegraphics[width=3in]{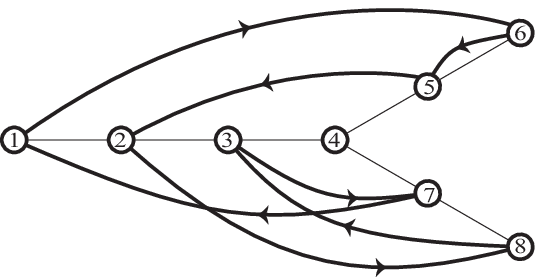}}

\subfloat[]{\includegraphics[width=3in]{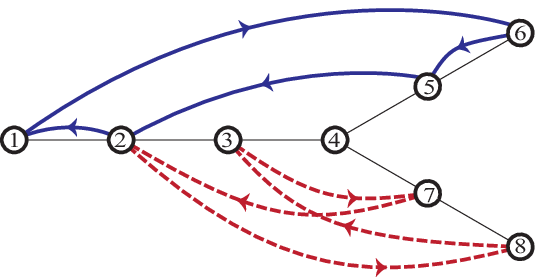}}
\end{center}
\caption{The cycle $(1\ 6\ 5\ 2\ 8\ 3\ 7)$ in Figure (a) is decomposed into two cycles, $  (1\ 6\ 5\ 2)$ and $ (2\ 8\ 3\ 7)$, as shown in Figure (b). Note that $(1\ 6\ 5\ 2\ 8\ 3\ 7)=(1\ 6\ 5\ 2)(2\ 8\ 3\ 7) $.}
\label{FigDecomposeTree}
\end{figure}


One may argue that the results of Lemma~\ref{lem:metric-path} and its extension to metric-trees have
limited application, as they require that both the defining tree and the permutations/rankings used in the 
computation be of special form. In particular, one may require that a given ranking $\pi$ is such that there are 
\emph{no edges} between two different branches of $\Theta$ in the cycle graph of $\pi$. We show next that under certain conditions the probability of such permutations goes to zero as $n \to \infty$, by lower bounding the number $P_{n}$ of permutations with
the given constraint.

Let the set of vertices in the $i$th
branch of a $Y$ shaped defining tree $\Theta$, $i=1,2,3,$ be denoted by $B_{i}$ and let $b_{i}$
denote the number of vertices in $B_{i}$. Clearly, $b_{1}+b_{2}+b_{3}+1=n$.

Assume, without loss of generality, that the numbering of the branches is such that $b_{1}\ge b_{2}\ge b_{3}$.
As an illustration, in Figure 3b we have 
\begin{align*}
B_{1} & =\{1,2,3\},\quad b_{1}=3,\\
B_{2} & =\{5,6\},\qquad b_{2}=2,\\
B_{3} & =\{7,8\},\qquad b_{3}=2.
\end{align*}

The quantity $P_{n}$ is greater than or equal to the number of permutations $\pi$
whose cycle decomposition does not contain an edge between $B_{2}$ and $B_{3}$,
and this quantity is, in turn, greater than or equal to the number of permutations $\pi$
such that $\pi\left(j\right)\notin B_{2}\cup B_{3}$ for $j\in B_{2}\cup B_{3}$. The number of
permutation with the latter property equals $\binom{b_{1}+1}{b_{2}+b_{3}}(b_{2}+b_{3})!(b_{1}+1)!$.
Hence,
\[
P_{n}\ge\frac{\left((b_{1}+1)!\right)^{2}}{(b_{1}+1-b_{2}-b_{3})!},
\]
and thus
\begin{align*}
\frac{P_{n}}{n!}\ge & \frac{\prod_{j=n+1-2b_{2}-2b_{3}}^{n-b_{2}-b_{3}}j}{\prod_{j=n+1-b_{2}-b_{3}}^{n}j}\cdot
\end{align*}

In particular, if $b_{2}=b_{3}=1$, we have 
\[
\frac{P_{n}}{n!}\ge\frac{(n-3)(n-2)}{(n-1)n}=1-\frac{4}{n}+O(n^{-2})
\]
 and more generally, if $b_{2}+b_{3}=o(n)$, then 
\[
\frac{P_{n}}{n!}\ge\frac{(n+o(n))^{b_{2}+b_{3}}}{(n+o(n))^{b_{2}+b_{3}}}\sim1,
\]
or equivalently, $P_{n}\sim n!$. 

Hence, if $b_{2}+b_{3}=o(n)$, the distance $\dist_{\varphi}(\pi,e)$ of a randomly chosen permutation $\pi$ from the identity equals $D_\varphi(\pi,e)/2$ with probability approaching 1 as $n\to\infty$.


It is worth noting that for metric-tree weight functions, the equality of Lemma~\ref{lem:metric-path} is not, in general, satisfied. To prove this claim, 
consider the metric-tree weight function $\wtfn$ in \figurename~\ref{fig:lemcounterexample}, where, for $a,b\in[n],a<b$,
\begin{equation*}
\nwt{a}b = \begin{cases}
1,&\qquad\text{if }a=1,\\
\infty,&\qquad\text{if }a\neq1.
\end{cases}
\end{equation*}It can be shown that for the permutation $\pi=(2\,3\,4)$, $\dist_\wtfn(\pi,e)= 4,$ while $\frac12D_\wtfn(\pi,e)=3$.
\begin{figure}
\begin{center}
\includegraphics[width=1.5in]{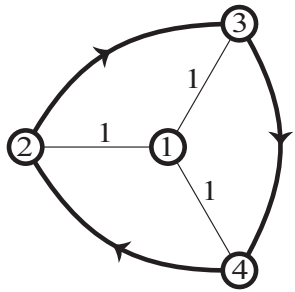}
\caption{For the above metric-tree weight function and $\pi=(2\,3\,4)$, the equality of Lemma~\ref{lem:metric-path} does not hold.} 
\label{fig:lemcounterexample}
\end{center}
\end{figure}

The following lemma provides a two approximation for transposition distances based on extended-path weight functions. 
As the weighted Kendall distance is a special case of the weighted transposition distance with extended-path weight functions, the lemma also implies Prop.~\ref{prop:wk2approximation}.
\begin{lem}
For an extended-path weight function $\wtfn$ and for $\pi,\sigma\in\S_n$,
\[\frac12 D_\wtfn(\pi,\sigma)\le\dist_\wtfn(\pi,\sigma) \le D_\wtfn(\pi,\sigma).\]
\end{lem}
\begin{proof}
The lower bound follows from Lemma~\ref{lem:LB}. 
To prove the upper bound, consider a metric-path weight function $\varphi',$ with the same defining path $\Theta$ as $\varphi,$ such that   
\[
\varphi'_{(a\ b)} = 2 \varphi_{(a\ b)}
\]
for any pair $a,b$ adjacent in $\Theta$. From Lemma~\ref{lemTranspositionUB}, it follows that for distinct $c,d\in[n]$,  
\[\dist_\varphi((c\ d),e)\le 2p^*_\varphi(c,d) = p^*_{\varphi'}(c,d) = \dist_{\varphi'}((c\ d),e).\]
 Hence,
\[\dist_\varphi(\pi,\sigma)\le\dist_{\varphi'}(\pi,\sigma)=\frac12D_{\wtfn'}(\pi,\sigma)=D_\wtfn(\pi,\sigma),\]
which proves the claimed result.
\end{proof}

\section{Aggregation Algorithms}
\label{sec:alg}

Despite the importance of the rank aggregation problem in many areas of information retrieval, only a handful of results regarding the complexity of the problem are known. Among them,
the most important results are the fact that finding a Kemeny optimal solution is NP-hard (see~\cite{dwork2001rank-web,popov} and references therein). Since the Kendall $\tau$ distance is a special case of the weighted Kendall distance, finding the aggregate ranking for the latter is also NP-hard. In particular, exhaustive search approaches -- akin to the one we used in the previous sections -- are not computationally feasible for large problems. 

However, assuming that $\pi^{*}$ is the solution to (\ref{eqn:rank-agg}), the ranking $\sigma_{i}$ closest to $\pi^{*}$ provides a 2-approximation for the aggregate ranking. 
This easily follows from the fact that the Kendall $\tau$ distance satisfies the triangle inequality. As a result, one only has to evaluate the pairwise distances of the votes $\Sigma$ in order to identify a 
2-approximation aggregate for the problem. Assuming the weighted Kendall distance can be computed efficiently (for example, if the weight function is monotonic), 
the same is true of the weighted Kendall distance as 
it is also a metric and thus satisfies the triangle inequality.

A second method for obtaining a 2-approximation is an extension of a bipartite matching algorithm. For any distance function that may be written as 
\begin{equation} \dist(\pi,\sigma)=\sum_{k=1}^{n}f(\pi^{-1}(k),\sigma^{-1}(k)),\label{eq:footrule}
\end{equation}
 where $f$ denotes an arbitrary non-negative function, one can find an \emph{exact solution} to (\ref{eqn:rank-agg}) as described in the next section. The matching algorithm approach
 for classical Kendall $\tau$ aggregation was first proposed in~\cite{dwork2001rank-web}. 
 
 \subsection{Vote Aggregation Using Matching Algorithms}
 
 Consider a complete weighted bipartite graph $\mathcal{G}=(X,Y)$, with $X=\{1,2,\cdots,n\}$ corresponding to the $n$ ranks to be filled in, and $Y=\{1,2,\cdots,n\}$ corresponding to the elements of $[n]$, i.e., the candidates. 
 Let $(i,j)$ denote an edge between $i\in X$ and $j\in Y$. We say that a perfect bipartite matching $P$ corresponds to a permutation $\pi$ whenever $(i,j)\in P$ if and only if $\pi(i)=j$. If the weight of $(i,j)$ equals 
 \begin{equation*} \sum_{l=1}^{m}f(i,\sigma_{l}^{-1}(j)),\label{eq:linear} \end{equation*}
 i.e., the weight incurred by $\pi(i)=j$, the minimum weight perfect matching corresponds to a solution of~(\ref{eqn:rank-agg}). The distance of \eqref{eq:footrule} is a generalized version of Spearman's footrule since Spearman's footrule \cite{diaconis1988group} can be obtained by choosing $f(x,y)=|x-y|$. 
 Below, we explain how to use the matching approach for aggregation based on a general weighted Kendall distance. 
 More details about this approach may be found in our companion conference paper~\cite{spcom2012}.

Recall that for a weighted Kendall distance with weight function $\varphi$,
\begin{equation*}
D_\wtfn(\pi,\sigma) = \sum_{i=1}^n w(\pi^{-1}(i):\sigma^{-1}(i)),
\end{equation*}
where 
\begin{equation*}
w(k:l) =  
\begin{cases}
\sum_{h=k}^{l-1}\nwt{h}{h+1}, & \text{if } k<l,\\
\sum_{h=l}^{k-1}\nwt{h}{h+1}, & \text{if } k>l,\\
0, & \text{if } k=l.
\end{cases}
\end{equation*}
Note that $D_\wtfn$ is a distance measure of the form of \eqref{eq:footrule}, and thus a solution to problem \eqref{eqn:rank-agg} for $\dist=D_\wtfn$ can be found 
exactly in polynomial time.

Suppose that the set of votes is given by $\Sigma=\{\sigma_1,\cdots,\sigma_m\}$.
\begin{prop} \label{propAggApprox}
Let $\pi'=\arg\min_{\pi}\sum_{l=1}^{m}D_\wtfn(\pi,\sigma_{i})$ and $\pi^{*}=\arg\min_{\pi}\sum_{l=1}^{m}\dist_\wtfn(\pi,\sigma_{i})$. The permutation $\pi'$ is a 2-approximation to the optimal rank aggregate $\pi^{*}$ if $\varphi$ corresponds to a weighted Kendall distance.\end{prop} 
\begin{proof} From Prop.~\ref{prop:wk2approximation}, for a weighted Kendall weight function $\wtfn$ and for permutations $\pi$ and $\sigma$, \[\frac12 D_\wtfn(\pi,\sigma) \le \dist_\wtfn(\pi,\sigma) \le D_\wtfn(\pi,\sigma).\]
Thus we have 
\[\sum_{l=1}^{m}\dist_{\varphi}(\pi',\sigma_{i})\le\sum_{l=1}^{m}D_\wtfn(\pi',\sigma_{i}). \] 
and 
\[ \frac12\sum_{l=1}^{m}D_\wtfn(\pi^{*},\sigma_{i})\le \sum_{l=1}^{m}\dist_{\varphi}(\pi^{*},\sigma_{i})\] 
From the optimality of $\pi'$ with respect to $D$, we find 
\[\sum_{l=1}^{m}D_\wtfn(\pi',\sigma_{i})\le \sum_{l=1}^{m}D_\wtfn(\pi^{*},\sigma_{i}).\] Hence
\[ \sum_{l=1}^{m}\dist_{\varphi}(\pi',\sigma_{i})\le2\sum_{l=1}^{m}\dist_{\varphi}(\pi^{*},\sigma_{i}). \]\end{proof}
In fact, the above proposition applies to the larger class of weighted transposition distances with extended-path weights. It can similarly be shown that for a weighted transposition distance with general weights (resp. metric weights), $\pi'$ is a 4-approximation (resp. a 2-approximation). Finally, for a weighted transposition distance with metric-path weights, $\pi'$ represents the exact solution.

A simple approach for improving the performance of the matching based algorithm is to couple it with a local descent method. 
Assume that an estimate of the aggregate at step $\ell$ equals $\pi^{(\ell)}$. As before, let $\mathbb A_n$ be the set of adjacent transpositions in $\S_n$. Then \[ \pi^{(\ell+1)}=\pi^{(\ell)}\,\arg\min_{\tau\in\mathbb A_n}\sum_{i=1}^{m}\dist(\pi^{(\ell)}\,\tau,\sigma_{i}). \] The search terminates when the cumulative distance of the aggregate from the set of votes $\Sigma$ cannot be decreased further. We choose the starting point $\pi^{(0)}$ to be the ranking $\pi'$ of Prop.~\ref{propAggApprox} obtained by the minimum weight bipartite matching algorithm. This method will henceforth be referred to as Bipartite Matching with Local Search (BMLS).

An important question at this point is how does the approximate nature of the BMLS aggregation process change the aggregate, especially with respect to the top-vs-bottom or similarity property? 
This question is hard, and we currently have no mathematical results pertaining to this problem. Instead, we describe a number of simulation results that may guide future analysis of this issue.

 In order to see the effect of the BMLS on vote aggregation, we revisit Examples \ref{example5x5}-\ref{example7x5}. In \emph{all except for one case} the solution provided by BMLS is the same as the exact solution, 
 both for the Kendall $\tau$ and weighted Kendall distances. 
 
 The exception is Example \ref{example5x4}. In this case, for the weight function $\nwt{i}{i+1}=(2/3)^{i-1},i\in[3]$, the exact solution equals $(1,4,2,3)$ but the solution obtained via BMLS equals $(4,2,3,1)$. Note that 
 these two solutions differ significantly in terms of their placement of candidate $1$, ranked first in the exact ranking and last in the approximate ranking. 
 The distances between the two solutions, $\dist_\wtfn((1,4,2,3),(4,2,3,1))$, equals $2.11$ and is rather large. Nevertheless, the cumulative distances to the votes are very close in value: 
\begin{equation*}
\begin{split}
\sum_i \dist_\wtfn((1,4,2,3),\sigma_i)&=9,\\
\sum_i \dist_\wtfn((4,2,3,1),\sigma_i)&=9.11.
\end{split}
\end{equation*}
Hence, as with any other distance based approach, the approximation result may sometimes diverge significantly from the optimum solution while the closeness of the approximate solution to the set of votes is nearly the same as that of the optimum solution. One way to avoid such approximation errors is to use weight functions with sufficiently large ``spreads'' of weights for which the difference between solutions has to be smaller than a given threshold. This topic will be discussed elsewhere.

\subsection{Vote Aggregation Using PageRank}


An algorithm for data fusion based on the PageRank and HITS algorithms for ranking web pages was proposed in~\cite{popov,dwork2001rank-web}. PageRank is one of the most important algorithms developed for search engines used by Google, with the aim of scoring web-pages based on their relevance. Each webpage that has hyperlinks to other webpages is considered a voter, while the voter's preferences for candidates is expressed via the hyperlinks. When a hyperlink to a webpage is not present, it is assumed that the voter does not support the given candidate's webpage. Although the exact implementation details of PageRank are not known, it is widely assumed that the graph of webpages is endowed with near-uniform transition probabilities. The ranking of the webpages is obtained by computing the stationary probabilities of the chain, and ordering the pages according to the values of the stationary probabilities. The connectivity of the Markov chain provides information about pairwise candidate preferences, and states with high input probability correspond to candidates ranked highly in a large number of lists. 

This idea can be easily adapted to the rank aggregation problem with weighted distances in several different settings. In such an adaptation, the states of a Markov chain correspond to the candidates and the transition probabilities are functions of the votes. Dwork et al.~\cite{dwork2001rank,dwork2001rank-web} proposed four different ways for computing the transition probabilities from the votes. Below, we describe the method that is most suitable for our problem and provide a generalization of the algorithm for weighted distance aggregation.


Consider a Markov chain with states indexed by the candidates. Let $P$ denote the transition probability matrix of the Markov chain, with $P_{ij}$ denoting the probability of going from state (candidate) $i$ to state $j$. 
In \cite{dwork2001rank-web}, the transition probabilities are evaluated as \[ P_{ij}=\frac{1}{m}\sum_{\sigma\in\Sigma}P_{ij}(\sigma), \] where
\[ P_{ij}(\sigma)=\begin{cases} \frac1n, & \qquad\mbox{if }\sigma^{-1}(j)<\sigma^{-1}(i),\\
1-\frac{\sigma^{-1}(i)-1}n, & \qquad\mbox{if }i=j,\\
 0, & \qquad\mbox{if }\sigma^{-1}(j)>\sigma^{-1}(i). \end{cases} \]

Our Markov chain model for weighted Kendall distance is similar, with a modification that includes incorporating transposition weights into the transition probabilities. To accomplish this task, we proceed as follows.

Let $w_{k}=\nwt{k}{k+1}$, and let $i_{\sigma}=\sigma^{-1}(i)$ for candidate $i\in[n]$.  We set \begin{equation} \beta_{ij}(\sigma)=\max_{l:j_{\sigma}\le l<i_{\sigma}}\frac{\sum_{h=l}^{i_\sigma-1}w_h}{i_{\sigma}-l}\label{eq:beta} \end{equation} if $j_{\sigma}<i_{\sigma}$, $\beta_{ij}(\sigma)=0$ if $j_{\sigma}>i_{\sigma}$, and \[ \beta_{ii}(\sigma)=\sum_{k:k_{\sigma}>i_{\sigma}}\beta_{ki}(\sigma). \]


The transition probabilities equal \[ P_{ij}=\frac{1}{m}\sum_{k=1}^{m}P_{ij}(\sigma_{k}), \] with \[ P_{ij}(\sigma)=\frac{\beta_{ij}(\sigma)}{\sum_{k}\beta_{ik}(\sigma)}. \] 

Intuitively, the transition probabilities described above may be interpreted as follows. The transition probabilities are obtained by averaging the transition probabilities corresponding to individual votes $\sigma\in\Sigma$. For each vote $\sigma$, consider candidates $j$ and $k$ with $j_\sigma=i_\sigma-1$ and $k_\sigma=i_\sigma-2$. The probability of going from candidate $i$ to candidate $j$ is proportional to $w_{j_{\sigma}}=\nwt{j_\sigma}{i_\sigma}$. This implies that if $w_{j_{\sigma}}>0$, one moves from candidate $i$ to candidate $j$ with positive probability. Furthermore, larger values for $w_{j_{\sigma}}$ result in higher probabilities for moving from $i$ to $j$.

In the case of candidate $k$, it seems reasonable to let the probability of transitioning from candidate $i$ to candidate $k$ be proportional to $\frac{w_{j_{\sigma}}+w_{k_{\sigma}}}{2}$. However, since $k$ is ranked before $j$ by vote $\sigma$, it is natural to require that the probability of moving to candidate $k$ from candidate $i$ be at least as high as the probability of moving to candidate $j$ from candidate $i$. This reasoning leads to $\beta_{ik}=\max\{w_{j_{\sigma}},\frac{w_{j_{\sigma}}+w_{k_{\sigma}}}{2}\}$ and motivates using the maximum in (\ref{eq:beta}). Finally, the probability of staying with candidate $i$ is proportional to the sum of the $\beta$'s from candidates placed below candidate $i$. 

\begin{figure}
\begin{center}
\includegraphics{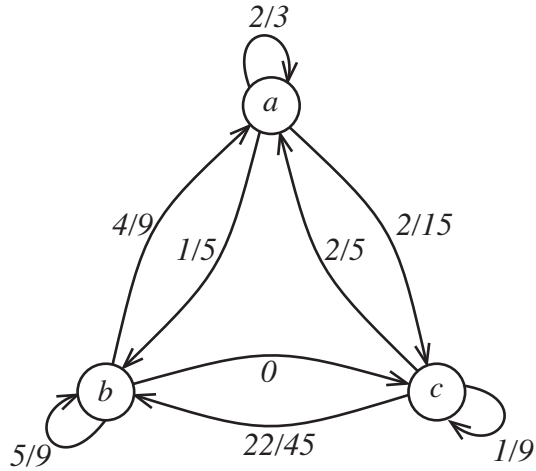}
\caption{The Markov chain for Example~\ref{ex:networkandmarkov}.}
\label{fig:networkandmarkov}
\end{center}
\end{figure}

\begin{example}\label{ex:networkandmarkov} 
Let the votes in $\Sigma$ consist of $\sigma_1 = (a,b,c)$, $\sigma_2=(a,b,c)$, and $\sigma_3 = (b,c,a)$, and let $w=(w_1,w_2) = (2,1)$. 

Consider the vote $\sigma_{1}=(a,b,c)$. We have $\beta_{ba}\left(\sigma_{1}\right)=\frac{w_{1}}{1}=2$. Note that if $w_{1}$ is large, then $\beta_{ba}$ is large as well.

In addition, $\beta_{cb}\left(\sigma_{1}\right)=\frac{w_{2}}{1}=1$ and \[\beta_{ca}=\max\left\{ \frac{w_{1}+w_{2}}{2},\beta_{cb}\right\} =\frac{3}{2}.\] The purpose of 
the $\max$ function is to ensure that $\beta_{ca}\ge\beta_{cb},$ which is a natural requirement given that $a$ is ranked before $b$ according to $\sigma_{1}$.

Finally, $\beta_{aa}\left(\sigma_{1}\right)=\beta_{ca}\left(\sigma_{1}\right)+\beta_{ba}\left(\sigma_{1}\right)=2+\frac{3}{2}=\frac{7}{2}$ and $\beta_{bb}\left(\sigma_{1}\right)=\beta_{cb}\left(\sigma_{1}\right)=1$. 
Note that according to the transition probability model, one also has $\beta_{aa}\ge\beta_{bb}.$ This may again be justified by the fact that $\sigma_{1}$ places $a$ higher than $b$. 

Since $\sigma_1 = \sigma_2$, we have
\[
P\left(\sigma_{1}\right)=P\left(\sigma_2\right)=
\left(\begin{array}{ccc}
1 & 0 & 0\\
\nicefrac{2}{3} & \nicefrac{1}{3} & 0\\
\nicefrac{3}{5} & \nicefrac{2}{5} & 0
\end{array}\right).
\]
Similar computations yield
\begin{eqnarray*}
\beta_{cb}\left(\sigma_{3}\right)=2, & \beta_{ac}\left(\sigma_{3}\right)=1, & \beta_{ab}\left(\sigma_{3}\right)=\frac{3}{2}\\
\beta_{aa}\left(\sigma_{3}\right)=0, & \beta_{bb}\left(\sigma_{3}\right)=2+\frac{3}{2}=\frac{7}{2}, & \beta_{cc}\left(\sigma_{3}\right)=1,
\end{eqnarray*}
and thus
\[
P\left(\sigma_{3}\right)=\left(\begin{array}{ccc}
0 & \nicefrac{3}{5} & \nicefrac{2}{5}\\
0 & 1 & 0\\
0 & \nicefrac{2}{3} & \nicefrac{1}{3}
\end{array}\right).
\]
From the $P\left(\sigma_{1}\right),P\left(\sigma_{2}\right),$ and
$P\left(\sigma_{3}\right)$, we obtain

\[
P=\frac{P\left(\sigma_{1}\right)+P\left(\sigma_{2}\right)+P\left(\sigma_{3}\right)}{3}
=\left(\begin{array}{ccc}
\nicefrac{2}{3} & \nicefrac{1}{5} & \nicefrac{2}{15}\\
\nicefrac{4}{9} & \nicefrac{5}{9} & 0\\
\nicefrac{2}{5} & \nicefrac{22}{45} & \nicefrac{1}{9}
\end{array}\right).
\]
The Markov chain corresponding to $P$ is given in Figure~\ref{fig:networkandmarkov}. The stationary distribution of this Markov chain is $(0.56657,\allowbreak 0.34844, 0.084986)$ which corresponds to the ranking $(a,b,c)$. 
\end{example}


\begin{example}
The performance of the Markov chain approach described above cannot be easily evaluated analytically, as is the case with any related aggregation algorithm proposed so far. 

We hence test the performance of the scheme on examples for which the optimal 
solutions are easy to evaluate numerically. For this purpose, in what follows, we consider a simple test example, with $m=11$. 
The set of votes (rankings) is listed below
\[\Sigma^T = \left(\begin{array}{c|c|c|c|c|c|c|c|c|c|c} 
1 & 1 & 1 & 2 & 2 & 3 & 3 & 4 & 4 & 5 & 5\\ 
2 & 2 & 2 & 3 & 3 & 2 & 2 & 2 & 2 & 2 & 2\\ 
3 & 3 & 3 & 4 & 4 & 4 & 4 & 5 & 5 & 3 & 3\\ 
4 & 4 & 4 & 5 & 5 & 5 & 5 & 3 & 3 & 4 & 4\\ 
5 & 5 & 5 & 1 & 1 & 1 & 1 & 1 & 1 & 1 & 1 
\end{array}\right). \] 

Note that due to the transpose operator, each column corresponds to a vote, e.g., $\sigma_{1}=\left(1,2,3,4,5\right)$. 

Let us consider candidates 1 and 2. 
Using the majority rule, one would arrive at the conclusion that candidate 1 should be the winner, given that 1 appears most often at the top of the list. Under a number of other aggregation rules, including Kemeny's rule and Borda's method, candidate 2 would be the winner.  
\begin{center} \begin{table*} \begin{centering} \begin{tabular}{|c|c|c|c|c|} \hline \multirow{2}{*}{Method} & \multicolumn{4}{c|}{Aggregate ranking and average distance}\tabularnewline \cline{2-5} & $w=\left(1,0,0,0\right)$  & $w=(1,1,1,1)$ & $w=\left(1,1,0,0\right)$  & $w=\left(0,1,0,0\right)$ \tabularnewline \hline \hline OPT  & $\left(\underline{1},4,3,2,5\right)$, 0.7273  & $\left(2,3,4,5,1\right),$ 2.3636  & $\left(\underline{2,3},4,5,1\right)$, 1.455  & $\left(\underline{3,2},5,4,1\right)$, 0.636 \tabularnewline \hline BMLS  & $\left(\underline{1},2,3,4,5\right)$, 0.7273  & $\left(2,3,4,5,1\right),$ 2.3636  & $\left(\underline{2,3},1,5,4\right)$, 1.455  & $\left(\underline{2,3},1,5,4\right)$, 0.636 \tabularnewline \hline MC  & $\left(\underline{1},2,5,4,3\right)$, 0.7273  & $\left(2,3,4,5,1\right),$ 2.3636  & $\left(\underline{2,1},3,4,5\right)$, 1.546 & $\left(\underline{2,3},1,4,5\right)$, 0.636 \tabularnewline \hline \end{tabular} \par\end{centering}
\caption{The aggregate rankings and the average distance of the aggregate ranking from the votes for different weight functions $w$.}
\label{tab1} \end{table*}
\end{center}

Our immediate goal is to see how different weighted distance based rank aggregation algorithms would position candidates $1$ and $2$. 
The numerical results regarding this example are presented in Table \ref{tab1}. In the table, OPT refers to an optimal solution which was found by exhaustive search, and MC refers to the Markov chain method. 

If the weight function is $w=(w_1,\cdots,w_4)=(1,0,0,0)$, where $w_i=\varphi_{(i\,i+1)}$, the optimal aggregate vote clearly corresponds to the plurality winner. That is, the winner is the candidate with most voters ranking him/her as the top candidate. A quick check of Table \ref{tab1} reveals that all three methods identify the winner correctly. Note that the ranks of candidates other than candidate 1 obtained by the different methods are different, however this does not affect the distance between the aggregate ranking and the votes.

The next weight function that we consider is the uniform weight function, $w=\left(1,1,1,1\right)$. This weight function corresponds to the conventional Kendall $\tau$ distance. As shown in Table \ref{tab1}, all three methods produce $\left(2,3,4,5,1\right),$ and the aggregates returned by BMLS and MC are optimum.

The weight function $w=\left(1,1,0,0\right)$ corresponds to \emph{ranking of the top 2} candidates. OPT and BMLS return $2,3$ as the top two candidates, both preferring $2$ to $3$. The MC method, however, returns $2,1$ as the top two candidates, with a preference for $2$ over $1$, and a suboptimal cumulative distance. It should be noted that this may be attributed to the fact 
the the MC method is not designed to only minimize the average distance: another important factor in determining the winners via the MC method is that winning against strong candidates ``makes one strong''. In this example, candidate 1 beats the strongest candidate, candidate 2, three times, while candidate 3 beats candidate 2 only twice and this fact seems to be the reason for the MC algorithm to prefer candidate 1 to candidate 3. Nevertheless, the stationary probabilities of candidates 1 and 3 obtained by the MC method are very close to each other, as the vector of probabilities is $(\underline{0.137},0.555,\underline{0.132},0.0883,0.0877)$.

The weight function $w=(0,1,0,0)$ corresponds to \emph{identifying the top 2} candidates -- i.e., it is not important which candidate is the first and which is the second. 
The OPT and BMLS identify $\left\{ 2,3\right\} $ as the top two candidates. 

The MC method returns the stationary probabilities $\left(0,1,0,0,0\right)$ which means that candidate 2 is an absorbing state in the Markov chain. 
This occurs because candidate 2 is ranked first or second by all voters. The existence of absorbing states is a drawback of the Markov chain methods. One solution is to remove 2 from the votes and re-apply MC. The MC method in this case results in the stationary distribution $\left(p\left(1\right),p\left(3\right),p\left(4\right),p\left(5\right)\right)=\left(0.273,0.364,0.182,0.182\right),$ which gives us the ranking $\left(3,1,4,5\right)$. Together with the fact that candidate 2 is the strongest candidate, we obtain the ranking $\left(2,3,1,4,5\right)$.
\end{example}

\section{Appendix}

\subsection{Computing the Weight Functions with Two Identical Non-zero Weights}

The goal is to find the weighted Kendall distance $\dist_\varphi(\pi,e)$ with the weight function of \eqref{eqTwoIdWf}, for an arbitrary $\pi\in\mathbb{S}_{n}$. 
For this purpose, let $R_1 = \{1,\cdots,a\}$, $ R_2 = \{a+1,\cdots,b\} $, and $R_3 = \{b+1,\cdots,n\},$ and define
\[
N_{ij}^\pi =|\{ k\in R_j:\pi^{-1}(k)\in R_i\}|, \quad i,j\in\{1,2,3\}.
\]
That is, $N_{ij}^\pi$ is the number of elements whose ranks in $\pi$ belong to the set $R_i$ and whose ranks in $e$ belong to the set $R_j$. A sequence of transpositions 
that transforms $\pi$ into $e$ moves the $N_{ij}^\pi$ elements of $\{ k\in R_j:\pi^{-1}(k)\in R_i\}$ from $R_i$ to $R_j$. Furthermore, note that any transposition that swaps two elements with ranks in the same region $R_i,i\in[3],$ has weight zero, while for any transposition $\tau_l$ that swaps an element ranked in $R_1$ with an element ranked in $R_2$ or swaps an element ranked in $R_2$ with an element ranked in $R_3$, we have $\dist_\varphi(\tau_l,e)=1$. 
\ignore{
\begin{figure*}
\begin{center}
\includegraphics[width=6in]{./weight-functions_4.eps}
\caption{The graph of a weight function with two non-zero weights $\varphi_{(a\ a+1)}$ and $\varphi_{(b\ b+1)}$ with $\varphi_{(a\ a+1)}=\varphi_{(b\ b+1)}=1$.} 
\label{fig:twonose}
\end{center}
\end{figure*}
}

It is straightforward to see that $\sum_j N_{ij}^\pi = \sum_j N_{ji}^\pi$. In particular, $N_{12}^\pi+N_{13}^\pi=N_{21}^\pi+N_{31}^\pi$ and $N_{31}^\pi+N_{32}^\pi=N_{13}^\pi+N_{23}^\pi$. 

We show next that 
\[
\dist_\varphi(\pi,e)=\begin{cases}
2N_{13}^{\pi}+N_{12}^{\pi}+N_{23}^{\pi}, & \quad \; \text{if} \; N_{21}^{\pi}\ge1 \text{ or } N_{23}^{\pi}\ge1,\\
2N_{13}^{\pi}+1, & \quad \; \text{if} \; N_{21}^{\pi}=N_{23}^{\pi}=0.
\end{cases}
\]

Note that, from Prop. \ref{prop:wk2approximation}, we have
\begin{equation}\label{eqTwoIdNonZeroLB}
\dist_\varphi(\pi,e)\ge\frac12 D_\varphi(\pi,e)=2N_{13}^{\pi}+N_{12}^{\pi}+N_{23}^{\pi}.
\end{equation}

Suppose that $N_{21}^{\pi}\ge1$ or $N_{23}^{\pi}\ge1$. We find a transposition $\tau_l$, with $\dist_\varphi(\tau_l,e)=1,$ such that $\pi'=\pi\tau_l$ satisfies $D_{\varphi}(\pi',e)=D_{\varphi}(\pi,e)-2,$ and at least one of the following conditions:
\begin{equation}
\begin{cases}
\quad N_{21}^{\pi'}\ge1,\\
\mbox{or}\\
\quad N_{23}^{\pi'}\ge1,\\
\mbox{or}\\
\quad \pi'=e.
\end{cases}\label{eq:condititons}
\end{equation}
Applying the same argument repeatedly, and using the triangle inequality, proves that $ \dist_\varphi(\pi,e)\le \frac12D_\varphi(\pi,e)$ if $N_{21}^{\pi}\ge1$ or $N_{23}^{\pi}\ge1$. This, along with \eqref{eqTwoIdNonZeroLB}, shows that $ \dist_\varphi(\pi,e)=\frac12D_\varphi(\pi,e)$ if $N_{21}^{\pi}\ge1$ or $N_{23}^{\pi}\ge1$.

First, suppose that $N^{\pi}_{21}\ge1$ and $N^{\pi}_{23}\ge1$. It then follows that $N^{\pi}_{12}\ge1$ or $N^{\pi}_{32}\ge1$. 
Without loss of generality, assume that $N^{\pi}_{12}\ge1$. Then $\tau_l$ can be chosen such that 
$N^{\pi'}_{12}=N^{\pi}_{12}-1$ and $N^{\pi'}_{21}=N^{\pi}_{21}-1$. We have $D_{\varphi}(\pi',e)=D_{\varphi}(\pi,e)-2,$ and since $N^{\pi}_{23}\ge1$, condition (\ref{eq:condititons}) holds.

Next, suppose $N^{\pi}_{21}\ge1$ and $N^{\pi}_{23}=0$. If $N^{\pi}_{13}\ge1$, choose $\tau_l$ such that
\begin{align*}
N^{\pi'}_{21} & =N^{\pi}_{21}-1,\\
N^{\pi'}_{23} & =1,\\
N^{\pi'}_{13} & =N^{\pi}_{13}-1,
\end{align*}
where $\pi'=\pi\tau_l$. Since $N^{\pi'}_{23}=1$, condition (\ref{eq:condititons}) is satisfied. 
If $N^{\pi}_{13}=0$, then $N^{\pi}_{31}=N^{\pi}_{32}=0,$ and thus $N^{\pi}_{12}=N^{\pi}_{21}\ge1$. In this case, we choose $\tau_l$ such that 
$N^{\pi'}_{21}=N^{\pi'}_{12}=N^{\pi}_{12}-1$. As a result, we have either $N^{\pi'}_{21} \ge1$ or $\pi'=e$. Hence, condition (\ref{eq:condititons}) is satisfied once again. 
Note that in both cases, for $N^{\pi}_{13}=0$ as well as for $N^{\pi}_{13}\ge1$, we have $D_{\varphi}(\pi',e)=D_{\varphi}(\pi,e)-2$.

The proof for the case $N^{\pi}_{23}\ge1$ and $N^{\pi}_{21}=0$ follows along similar lines.

If $N^\pi_{21}=N^\pi_{23}=0$, it can be verified by inspection that for every transposition $\tau_l$ with $\dist_\varphi(\tau_l,e)=1$, we have $D_\varphi(\pi\tau_l,e)\ge D_\varphi(\pi,e)$. Hence, the inequality in \eqref{eqTwoIdNonZeroLB} cannot be satisfied with equality, which implies that $\dist_\varphi(\pi,e)\ge 2N_{13}^\pi+1$. Choose a transposition $\tau_l$ with $\dist_\varphi(\tau_l,e)=1$ such that 
\begin{align*}
N^{\pi'}_{13} & =N^{\pi}_{13}-1,\\
N^{\pi'}_{12} & =1,\\
N^{\pi'}_{23} & =1.
\end{align*}
where $\pi'=\pi\tau_l$. We have 
\begin{equation*}
\begin{split}
\dist_\varphi(\pi,e)\le\dist_\varphi(\tau_l,e)+\dist_\varphi(\pi',e)=1+2N^{\pi}_{13}.
\end{split}
\end{equation*}
This, along with $\dist_\varphi(\pi,e)\ge 2N_{13}^\pi+1$, completes the proof.

\subsection{The Average Kendall and Weighted Kendall Distance}

The Kendall $\tau$ distance between two rankings may be viewed in the following way: each pair of candidates on which the two rankings disagree contribute one unit to the distance between the rankings. 
Owing to Algorithm \ref{alg:FindTauMonotone}, the weighted Kendall distance with a decreasing weight function can be regarded in a similar manner: each pair of candidates on which the two rankings disagree contributes $\varphi_{(s\ s+1)}$, for some $s$, to the distance between the rankings. 

Consider a pair $a$ and $b$ such that $\pi^{-1}(b)<\pi^{-1}(a)$ and $\sigma^{-1}(a)<\sigma^{-1}(b)$. In Algorithm~\ref{alg:FindTauMonotone}, there exists a transposition $\tau_{t}^{\star}=(s\ s+1)$ that swaps $a$ and $b$ where 
\[
s=\pi^{-1}(b)+\left|\left\{ k:\sigma^{-1}(k)<\sigma^{-1}(a), \pi^{-1}(k)>\pi^{-1}(b)\right\} \right|,
\]
that is, $s$ equals $\pi^{-1}(b)$ plus the number of elements that appear before $a$ in $\sigma$ and after $b$ in $\pi$. It is not hard to see that $s$ can also be written in a way that is symmetric with respect to $\pi$ and $\sigma,$ as
\begin{equation}
\label{eq:whichweight}
\begin{split}
s&
=\pi^{-1}(b)+\sigma^{-1}(a)-\left|\left\{k:\pi^{-1}(k)<\pi^{-1}(b),\sigma^{-1}(k)<\sigma^{-1}(a)\right\}\right|-1\nonumber\\
&=n-1-\left|\left\{ k:\pi^{-1}(k)>\pi^{-1}(b),\sigma^{-1}(k)>\sigma^{-1}(a)\right\}\right|.
\end{split}
\end{equation}


As an example, consider $\varphi_{(i\ i+1)} = n-i$. Then,
\begin{align*}
\dist_\varphi(\pi,\sigma) &= \sum_{(b,a)\in \mathscr I(\pi,\sigma)} \left(1+\left|\left\{ k:\pi^{-1}(k)>\pi^{-1}(b),\sigma^{-1}(k)>\sigma^{-1}(a)\right\}\right|\right)\\
&= \kdist(\pi,\sigma) + \sum_{(b,a)\in \mathscr I(\pi,\sigma)} \left|\left\{ k:\pi^{-1}(k)>\pi^{-1}(b),\sigma^{-1}(k)>\sigma^{-1}(a)\right\}\right|
\end{align*}
where $\mathscr I(\pi,\sigma)$ is the set of ordered pairs $(b,a)$ such that $\pi^{-1}(b)<\pi^{-1}(a)$ and $\sigma^{-1}(a)<\sigma^{-1}(b)$. Note that the weighted Kendall distance $\dist_\varphi$ equals the Kendall $\tau$ distance plus a sum that captures the influence of assigning higher importance to the top positions of the rankings. 

These observations allow us to easily compute the expected value of the distance between the identity permutation and a randomly and uniformly chosen permutation $\pi \in \mathbb{S}_{n}$. 
For $1\le a<b\le n$ and $s\in[n-1]$, let $X_{ab}^{s}$ be an indicator variable that equals one if and only if $\pi^{-1}(a)>\pi^{-1}(b)$ and
\[
\left|\left\{ k>a:\pi^{-1}(k)>\pi^{-1}(b)\right\} \right|=n-1-s.
\]
The expected distance between the two permutations equals
\begin{equation}
E[\dist_\varphi(\pi,e)]=\sum_{s=1}^{n-1}\varphi_{(s \ s+1)}  \sum_{a=1}^{n-1}\sum_{b=a+1}^{n}E\left[X_{ab}^{s}\right].\label{eq:expectation}
\end{equation}
By the definition of $X_{ab}^s$, $E\left[X_{ab}^{s}\right]$ equals the probability of the event that $n-1-s$ elements of $\{a+1,\cdots,n\}\backslash\{b\}$ and $a$ appear after $b$ in $\pi$. There are $\binom{n-a-1}{n-s-1}$ ways to choose $n-s-1$ elements from $\{a+1,\cdots,n\}\backslash\{b\}$, $\binom{n}{a-1}(a-1)!$ ways to assign positions to the elements of $\{1,2,\cdots,a-1\}$, $(s-a)!$ ways to arrange the $s-a$ elements of $\{a+1,\cdots,n\}\backslash\{b\}$ that appear before $b$, and $(n-s)!$ ways to arrange $a$ and the $n-1-s$ elements of $\{a+1,\cdots,n\}\backslash\{b\}$ that appear after $b$. Hence,
\begin{align*}
E\left[X_{ab}^{s}\right] & =\frac{1}{n!}\binom{n-a-1}{n-s-1}\binom{n}{a-1}(a-1)!(s-a)!(n-s)!\\
 & =\frac{n-s}{(n-a+1)(n-a)},
\end{align*}
for $1\le a\le s,$ and $E\left[X_{ab}^{ s}\right]=0$ for
$a> s$. Using this expression in (\ref{eq:expectation}), we obtain
\begin{align*}
E[\dist_\varphi(e,\pi)] & =\sum_{ s=1}^{n-1}\varphi_{(s\ s+1)}\sum_{a=1}^{ s}\frac{n- s}{n-a+1}\\
 & =\sum_{ s=1}^{n-1}\varphi_{(s\ s+1)}(n- s)(H_{n}-H_{n-s}),
\end{align*}
where $H_i = \sum_{l=1}^i \frac1l$. Indeed, for $\varphi_{(s\,s+1)}=1,s\in[n-1]$, we recover the well known result that
\begin{align*}
E[\dist_\varphi(e,\pi)] & =\sum_{ s=1}^{n-1}(n- s)(H_{n}-H_{n-s})\\
& =\sum_{k=1}^{n-1}k(H_{n}-H_k)\\
& = \frac12\binom{n}{2}.
\end{align*}
For $\varphi_{(s\ s+1)}=n-s$, the average distance equals
\begin{align*}
E[\dist_\varphi(e,\pi)] & =\sum_{ s=1}^{n-1}(n- s)^2(H_{n}-H_{n-s})\\
& =\sum_{k=1}^{n-1}k^2(H_{n}-H_k)\\
& =\frac12\binom{n}{2}+\frac23\binom{n}{3}.
\end{align*}

\textbf{Acknowledgment}: The authors are grateful to 
Tzu-Yueh Tseng for helping with the numerical results and to 
Eitan Yaakobi and Michael Landberg for useful discussions. 

\bibliographystyle{abbrv} \bibliography{bib}

\begin{thebibliography}{10}

\bibitem{arrow1963social}
K.~J. Arrow.
\newblock {\em Social choice and individual values}.
\newblock Yale Univ Pr, 1963.

\bibitem{cook1985ordinal}
W.~D. Cook and M.~Kress.
\newblock Ordinal ranking with intensity of preference.
\newblock {\em Management Science}, 31(1):26--32, 01 1985.

\bibitem{cormen24introduction}
T.~Cormen, C.~Leiserson, R.~Rivest, and C.~Stein.
\newblock {\em {Introduction to algorithms}}.
\newblock MIT, Cambridge, MA, 2001.

\bibitem{borda1784}
J.-C. de~Borda.
\newblock M\'{e}moire sur les \'{e}lections au scrutin.
\newblock {\em Histoire de l'Acad\'{e}mie royale des sciences}, 1784.

\bibitem{condorcet}
M.~de~Condorcet.
\newblock {\em Paris: Imprimerie royale}, (1):14--37, Feb.

\bibitem{deza1998metrics}
M.~Deza and T.~Huang.
\newblock Metrics on permutations, a survey.
\newblock {\em J. Comb. Inf. Sys. Sci}, 23:173--185, 1998.

\bibitem{diaconis1988group}
P.~Diaconis.
\newblock Group representations in probability and statistics.
\newblock {\em Lecture Notes-Monograph Series}, 11, 1988.

\bibitem{dwork2001rank-web}
C.~Dwork, R.~Kumar, M.~Naor, and D.~Sivakumar.
\newblock Rank aggregation methods for the web.
\newblock In {\em Proceedings of the 10th international conference on World
  Wide Web}, pages 613--622. ACM, 2001.

\bibitem{dwork2001rank}
C.~Dwork, R.~Kumar, M.~Naor, and D.~Sivakumar.
\newblock Rank aggregation revisited.
\newblock Manuscript, Available:
  \url{www.eecs.harvard.edu/~michaelm/CS222/rank2.pdf}, 2001.

\bibitem{eiuliveability}
{Economist Intelligence Unit}.
\newblock {A Summary of the Liveability Ranking and Overview [white paper]}.
\newblock
  \url{http://www.eiu.com/site_info.asp?info_name=The_Global_Liveability_Report},
  Aug. 2012.

\bibitem{infretrieval}
M.~Farah and D.~Vanderpoten.
\newblock An outranking approach for rank aggregation in information retrieval.
\newblock {\em SIGIRÕ07}, pages 14--37, July 2007.

\bibitem{farnoud2012sorting}
F.~Farnoud and O.~Milenkovic.
\newblock Sorting of permutations by cost-constrained transpositions.
\newblock {\em {IEEE} Trans. Information Theory}, 58(1):3 --23, Jan. 2012.

\bibitem{spcom2012}
F.~Farnoud, B.~Touri, and O.~Milenkovic.
\newblock Nonuniform vote aggregation algorithms.
\newblock In {\em Int. Conf. Signal Processing and Communications}, India, July
  2012.

\bibitem{fredman:1987:dijkstra}
M.~L. Fredman and R.~E. Tarjan.
\newblock Fibonacci heaps and their uses in improved network optimization
  algorithms.
\newblock {\em J. ACM}, 34(3):596--615, July 1987.

\bibitem{hodge2005mathematics}
J.~Hodge and R.~Klima.
\newblock {\em The mathematics of voting and elections: a hands-on approach},
  volume~22 of {\em Mathematical World}.
\newblock American Mathematical Society, 2005.

\bibitem{bruck2009rank-modulation}
A.~Jiang, R.~Mateescu, M.~Schwartz, and J.~Bruck.
\newblock Rank modulation for flash memories.
\newblock {\em Information Theory, IEEE Transactions on}, 55(6):2659 --2673,
  June 2009.

\bibitem{kemeney1959mathematics}
J.~G. Kemeny.
\newblock Mathematics without numbers.
\newblock {\em Daedalus}, 88(4):pp. 577--591, 1959.

\bibitem{kemenySnell1962mathematical}
J.~G. Kemeny and J.~Snell.
\newblock {\em Mathematical models in the social sciences}.
\newblock Ginn, Boston, 1962.

\bibitem{kendall1970rankcorrelation}
M.~Kendall.
\newblock {\em Rank correlation methods}.
\newblock London: Griffin, 4th edition, 1970.

\bibitem{kumar2010gdr}
R.~Kumar and S.~Vassilvitskii.
\newblock Generalized distances between rankings.
\newblock In {\em Proceedings of the 19th international conference on World
  wide web}, WWW '10, pages 571--580, New York, NY, USA, 2010. ACM.

\bibitem{6034261}
A.~Mazumdar, A.~Barg, and G.~Zemor.
\newblock Constructions of rank modulation codes.
\newblock In {\em IEEE International Symposium on Information Theory}, pages
  869 --873, July/Aug. 2011.

\bibitem{optifyinc}
{Optify Inc.}
\newblock {The Changing Face of SERPs: Organic Click Through Rate [white
  paper]}.
\newblock
  \url{http://www.optify.net/inbound-marketing-resources/new-study-how-the-new-face-of-serps-has-altered-the-ctr-curve},
  2012.

\bibitem{popov}
V.~Popov.
\newblock Multiple genome rearrangement by swaps and by element duplications.
\newblock {\em Theoretical Computer Science}, 385(1-3):115--126, 2007.

\bibitem{schalekamp2009rank}
F.~Schalekamp and A.~van Zuylen.
\newblock Rank aggregation: Together we're strong.
\newblock {\em Proc. of 11th ALENEX}, pages 38--51, 2009.

\bibitem{sculley2007rank}
D.~Sculley.
\newblock Rank aggregation for similar items.
\newblock In {\em Proceedings of the Seventh SIAM International Conference on
  Data Mining}, 2007.

\end{thebibliography}

\end{document}